\pdfoutput=1

\documentclass[10pt,twocolumn,english,letterpaper,notitlepage]{article}
\usepackage[utf8]{inputenc}
\usepackage[T1]{fontenc}
\usepackage[letterpaper]{geometry}
\usepackage{color}
\usepackage{babel}
\usepackage{verbatim}
\usepackage{enumitem}
\usepackage{amsmath}
\usepackage{amsthm}
\usepackage{amssymb}
\usepackage[unicode=true,
 bookmarks=true,bookmarksnumbered=false,bookmarksopen=false,
 breaklinks=true,pdfborder={0 0 0},pdfborderstyle={},backref=false,colorlinks=true]
 {hyperref}
\usepackage[affil-it]{authblk}
\usepackage{abstract}
\usepackage{tensor}
\usepackage{marginnote}

\makeatletter
\theoremstyle{plain}
\newtheorem{thm}{\protect\theoremname}
\theoremstyle{plain}
\newtheorem{prop}[thm]{\protect\propositionname}
\theoremstyle{remark}
\newtheorem*{rem*}{\protect\remarkname}
\theoremstyle{plain}
\newtheorem{lem}[thm]{\protect\lemmaname}
\providecommand{\lemmaname}{Lemma}
\providecommand{\propositionname}{Proposition}
\providecommand{\remarkname}{Remark}
\providecommand{\theoremname}{Theorem}

\usepackage{titlesec} %change title spacing for sections
\titlespacing*{\section}
{0pt}{0.7\baselineskip}{0.3\baselineskip}
\titlespacing*{\subsection}
{0pt}{0.7\baselineskip}{0.3\baselineskip}

\renewenvironment{align*}{\align}{\endalign}
\renewenvironment{alignat*}{\alignat}{\endalignat}
\renewenvironment{multline*}{\multline}{\endmultline}
\renewenvironment{gather*}{\gather}{\endgather}
\renewenvironment{equation*}{\equation}{\endequation}

\renewcommand{\geq}{\geqslant}
\renewcommand{\leq}{\leqslant}

\renewcommand{\Re}{\operatorname{Re}}

\newcommand{\dd}{{\rm{d}}}

\DeclareMathOperator{\tr}{tr}

\usepackage{csquotes}

\AtBeginDocument{

\renewbibmacro*{doi+eprint+url}{%
  \iftoggle{bbx:doi}
    {\iffieldundef{url}{\printfield{doi}}{}}
    {}%
  \newunit\newblock
  \iftoggle{bbx:eprint}
    {\usebibmacro{eprint}}
    {}%
  \newunit\newblock
  \iftoggle{bbx:url}
    {\usebibmacro{url+urldate}}
    {}}

\newbibmacro{string+doiurlisbn}[1]{%
  \iffieldundef{url}{%
    \iffieldundef{doi}{%
     \iffieldundef{eprint}{%
      \iffieldundef{isbn}{%
        \iffieldundef{issn}{%
          #1%
        }{%
          \href{http://books.google.com/books?vid=ISSN\thefield{issn}}{#1}%
        }%
      }{%
        \href{http://books.google.com/books?vid=ISBN\thefield{isbn}}{#1}%
      }%
     }{
        \href{\thefield{eprint}}{#1}
     }%
    }{%
      \href{http://dx.doi.org/\thefield{doi}}{#1}%
    }%
  }{%
    \href{\thefield{url}}{#1}%
  }%
}

\renewbibmacro*{journal}{%
  \ifboolexpr{ test {\ifcitation} and not test {\iffieldundef{shortjournal}} }
    {\printfield[journaltitle]{shortjournal}}
    {\iffieldundef{journaltitle}
       {}
       {\printtext[journaltitle]{%
          \printfield[titlecase]{journaltitle}%
          \setunit{\subtitlepunct}%
          \printfield[titlecase]{journalsubtitle}}}}}

\renewbibmacro*{journal+sub}{%
  \ifboolexpr{ test {\ifcitation} and not test {\iffieldundef{shortjournal}} }
    {\printfield[journaltitle]{shortjournal}}
    {\iffieldundef{journaltitle}
       {}
       {\printtext[journaltitle]{%
        \printfield[noformat]{journaltitle}%
        \setunit{\addcolon\addspace}%
        \printfield[noformat]{journalsubtitle}}}}}

\DeclareSourcemap{%
  \maps[datatype=bibtex]{
    \map[overwrite]{
      \step[fieldsource=shortjournal]
      \step[fieldset=journaltitle,origfieldval]
    }
  }
}

\renewbibmacro{in:}{\ifentrytype{article}{}{\printtext{\bibstring{in}\intitlepunct}}}

\DefineBibliographyStrings{english}{%
    backrefpage  = {\lowercase{s}ee p.}, % for single page number
    backrefpages = {\lowercase{s}ee pp.} % for multiple page numbers
}

\DeclareFieldFormat{title}{\usebibmacro{string+doiurlisbn}{#1}}
\DeclareFieldFormat[article,incollection]{title}{\usebibmacro{string+doiurlisbn}{\mkbibquote{#1}}}
\DeclareFieldFormat{journaltitle}{\mkbibemph{#1},} % italic journal title with comma
\DeclareFieldFormat[inbook,thesis]{title}{\mkbibemph{#1}\addperiod} % italic title with period
\DeclareFieldFormat[article]{volume}{\textbf{#1}} % makes volume of journal bold and adds colon
\DeclareFieldFormat{pages}{#1} % removes pagination (p./pp.) before page numbers

\usepackage{setspace} % to change line spacing in bibliography with \setstretch

} % end of AtBeginDocument{...}

\usepackage{newtxtext}
\usepackage[varvw]{newtxmath}

\usepackage{bm}

\usepackage[activate={true,nocompatibility},final,tracking=true,kerning=true,spacing=true,factor=1100,stretch=20,shrink=20]{microtype}%text appearance improvements
\UseMicrotypeSet[protrusion]{basictext}
\microtypecontext{spacing=nonfrench}
\DisableLigatures[N]{encoding=*, family =*}
\SetTracking{encoding=*, shape=sc}{0} %reduce spacing in small caps
\DeclareMathSizes{12}{12}{8}{6}

\usepackage{array}% Enhancements for {array} and {tabular}
\newcolumntype{C}{>{$\displaystyle} c <{$}} 

\usepackage{xparse}
\DeclareDocumentCommand{\restr}{m m o}  
{%
  \IfNoValueTF{#3}
  {
  \left.\kern-\nulldelimiterspace % automatically resize the bar with \right
  #1 % the function
  \vphantom{\big|} % pretend it's a little taller at normal size
  \right|_{#2}
  }{
  \left.\kern-\nulldelimiterspace % automatically resize the bar with \right
  #1 % the function
  \vphantom{\big|} % pretend it's a little taller at normal size
  \right|_{#2}^{#3} }
}

\usepackage{graphicx}
\newcommand{\PRLsep}{   %decorative line used to separate sections
           \noindent\makebox[\linewidth]{
                \resizebox{0.7\linewidth}{1pt}{$\blacklozenge$}}\medbreak}

\usepackage{graphicx,scalerel}
\newcommand\wh[1]{\hstretch{1.42}{\hat{\hstretch{.7}{#1\mkern2mu}}}\mkern-2mu} %requires \usepackage{graphicx,scalerel}
\newcommand\wt[1]{\hstretch{1.42}{\tilde{\hstretch{.7}{#1\mkern2mu}}}\mkern-2mu}
\newcommand\wb[1]{\hstretch{1.42}{\bar{\hstretch{.7}{#1\mkern2mu}}}\mkern-2mu}

\makeatother

\usepackage[bibstyle=custom-numeric-comp,citestyle=numeric,sorting=nyt,backref=false,sortcites=true,firstinits=true,language=english,maxbibnames=99,isbn=false,url=false,doi=false,eprint=false]{biblatex}

\title{\vskip -3em Vortex flows on closed surfaces}
\author{\textsc{A. Bogatskiy}\thanks{Kadanoff Center for Theoretical Physics, University of Chicago, Chicago, IL, USA, 60637}}
\date{\vspace{-5ex}}

\begin{document}

\vskip -2em
\twocolumn[
  \begin{@twocolumnfalse}
    \maketitle
    \begin{abstract}
    	\centering\begin{minipage}{\dimexpr\paperwidth-10cm}
    		We investigate the bulk hydrodynamics of the chiral vortex matter on an arbitrary closed surface, extending the ideas of \cites{Khalatniko89}{WiegmAban14}. Placing this important example of a chiral medium onto a curved geometry reveals the geometric nature of \textit{odd viscosity}. The anomalous odd viscosity of the vortex matter is associated with a special interaction of point vortices with curvature.
    	\end{minipage}
    \end{abstract}
  \end{@twocolumnfalse}
]
\saythanks

\section{Introduction}

There is abundant literature on the $N$-vortex problem of ideal incompressible
2D fluids at fixed $N$ (e.g.\@ see \cite{AreRotTho92} on the integrability of the 3-vortex case), including an extended
search for sophisticated equilibrium solutions consisting of discrete vortices
and/or patches of constant vorticity (see e.g. \cite{CrowdClok02}). Some of the more recent research has been focused on the collective motion of a large number $N$ of vortices in the limit when the ratio $N/V$ of the particle number to the total area remains finite. This limit is especially interesting since the dynamics of $N>3$ vortices  is  believed to be chaotic \cite{Aref83,ArefPomp82}. Especially interesting flows are those consisting of identical sign-like vortices, the \textit{chiral} ``turbulent Euler flow'. This appears to be a rather recent and unexplored subject \cites{CagLiMaPu92}{CheKolZhi13}{WiegmAban14}{BogatWiegm19}.

It is natural to describe the chiral flow, a.k.a.\ \textit{vortex matter}, in the coarse-grained, or hydrodynamic, limit where vortices themselves constitute a new fluid. This is the limit when $N\to\infty$, the vortex number density $\rho$ remains bounded so that it can be replaced by a regular positive function with
$\int\rho\,\dd V=N$, and the gradients of the density are kept small $\vert\nabla\rho/\rho\vert\ll
\sqrt{\rho}$. In this case the velocities $v^\mu_i$ of the discrete vortices can be approximated by a continuous vector field $v^\mu$. Such flows are important in many contexts. A classical example of a chiral vortex flow is the rotating superfluid Helium \cite{Khalatniko89}, but there are others, e.g.\ Onsager's vortex clusters forming after an inverse cascade in confined fluids, the Bose-Einstein condensate,
etc. A coarse-grained description, albeit a heuristic one, of a rotating superfluid was already known in the 1960s \cite{Khalatniko89}. The main claim is that the coarse-grained energy contains a correction to the naive kinetic energy:
\[
E=\int\left[ \frac{u^2}{2}-\frac{\Gamma^2}{8\pi}\rho\ln\rho\right] \dd V,
\]
where $u$ is the ``coarse-grained flow'', i.e.\ the one with vorticity $\omega=\nabla\times u=\Gamma\rho$, and $\Gamma$ is the circulation of each vortex. A complete derivation for a rotating incompressible fluid in the infinite plane was first given in \cite{WiegmAban14}.
This new energy term, despite being a Casimir, determines a peculiar force acting within the vortex matter, the so-called \textit{odd viscosity}. The effect of the odd viscosity can be observed in the velocity of the vortex fluid $v$: while in a continuous flow the vorticity is convected by the flow, in the vortex matter it is not. The vortex flow deviates from the coarse-grained velocity $u$ as
\[
v^\mu=u^\mu-\eta\epsilon^{\mu\nu}\nabla_\nu\ln\rho,\quad \eta=\frac{\Gamma}{8\pi}.\label{WA formula}
\]
Here $\eta$ is the odd viscosity coefficient. The anomalous term originates from the discreteness of vortices: the velocity of the flow is infinite at the position of the vortex, so it is not clear what it means for vortices to be convected by the flow. The anomalous term reflects the regularization of the singularities in the vortex cores.

The anomalous force described above has a geometric nature. For this reason it is important to extend the results obtained for unbounded fluids to confined geometries where the vortex matter interacts with boundaries and with local geometry (like curvature). In this paper we develop a framework for this task and apply it to the vortex matter on closed surfaces. The main question is: how does odd viscosity manifest itself on curved surfaces, and how is (\ref{WA formula}) generalized? The answers are in the coarse-grained velocity (\ref{eq: v in terms of u (m)}) and energy (\ref{eq: coarse-grained H(m)}) of the vortex fluid. These also imply the  \textit{mean-field equation} (\ref{eq: Liouville-type
equation}) for stationary distributions of vortices. We follow a strategy vaguely similar to that of \cite{WiegmAban14}: we compute the stress tensor of the vortex matter and use it to coarse-grain the dynamics.

We now review vortices in the infinite plane. There vortices of strengths $\Gamma_{i}$ follow Kirchhoff's equations \cite{Kirchhoff76}. In complex coordinates $z_{j}=x_{j}+iy_{j}$ they read
\begin{equation}
i\dot{\bar{z}}_{i}=\frac{1}{2\pi}\sum_{j\neq i}\frac{\Gamma_{j}}{z_{i}-z_{j}}.\label{eq: Kirchhoff}
\end{equation}
 We are interested in the chiral case, when all vortices are the same,
$\Gamma_{j}=\Gamma$. In this case the total angular impulse of vortices
$L=\frac{i}{2}\sum_{i}\left(z_{i}\dot{\bar{z}}_{i}-\dot{z}_{i}\bar{z}_{i}\right)$ is conserved \cite{Saffman92}. At fixed $L$ and large $N$ the minimum energy configuration is a circular droplet of vortex matter of asymptotically uniform density in the bulk.
 
For smooth flows the Euler equation is minimally coupled to a metric. Defining point vortices as weak solutions of the Euler equation requires a resolution of singularities at the vortex cores. On curved or bounded domains this is especially difficult due to the absence of translational symmetry. While mathematical results in this context are limited \cite{Turkington87}, all existing methods lead to the classic Kirchhoff-Routh equations that were established by 1941 \cites{Lin41}{Lin41a}. A review can be found in \cite{FluchGusta97}. The Kirchhoff-Routh equations generalized to closed surfaces are the starting point of our work.

The Kirchhoff-Routh equations, originally written for fluids in planar domains, can be extended to curved surfaces as follows. If the infinite plane is deformed so that its metric becomes $\sqrt{g}\mathrm{d}z\mathrm{d}\bar{z}$, the Kirchhof-Routh equations
are  \begin{equation}
\sqrt{g\left(z_{j}\right)}i\dot{\bar{z}}_{j}=\frac{1}{2\pi}\sum_{k\neq j}\frac{\Gamma_{k}}{z_{j}-z_{k}}+\frac{\Gamma_{j}}{4\pi}\partial_{z_{j}}\ln\sqrt{g\left(z_{j}\right)}\label{eq:
Hally}
\end{equation}
(they were obtained heuristically in \cite{Hally80} using the ideas
of Routh \cite{Routh81}). The last term is the effect of the regularized self-energy of the vortex core, which can be interpreted as the condition that the area
of the vortex core be independent of its position on the surface. These equations have since been generalized to arbitrary surfaces \cite{BoattKoill08}, and we review them again below.

Vortex matter on surfaces with boundaries develops a separate boundary with sophisticated dynamics \cite{BogatWiegm19}. To avoid the complications caused by the boundary we focus on vortices on closed surfaces.

%For vortices in the plane and the closely related ``aggregation model''\mn{\blue stationary aggregation is equivalent to stationary vortices (the difference is only in the coefficient in front of z-dot), so their results are relevant) (where the factor of $i$ in (\ref{eq: Kirchhoff}) is dropped), this was done in \cite{FetHuaKol11,BerLauLeg12,CheKolZhi13,WiegmAban14}: the limiting density of stationary vortices is constant inside a circle of radius proportional to $\sqrt{\frac{\Gamma N}{\Omega}}$ and vanishing outside of it. The stationary solutions of Kirchhoff's equations are also equivalent to the zero temperature limit of the 2D Dyson gas \cite{ZabroWieg06}. Namely, vortices organize into domains of constant density with a complicated ``overshoot'' at the edge.

The paper is organized as follows: after introducing the necessary notation we define point vortex solutions on arbitrary closed surfaces. Next, to emphasize the geometric nature of the problem we introduce a new geometric force which characterizes a non-minimal coupling of the Euler and the Kirchhoff equations to the metric. Then we compute the stress of the vortex matter by varying the Kirchhoff-Routh energy w.r.t.\ the metric and relate it to the odd viscosity. We notice that the coarse-graining problem becomes trivial for the stress tensor. Taking the coarse-grained limit of the stress we then obtain the coarse-grained energy. This yields the hydrodynamics of the vortex matter. In the Appendices, we provide more information about the general behavior of vortices on surfaces, and give detailed derivations of the main claims of the paper.

\section{Background\label{subsec:Notation}}

\subsection{Green functions on closed surfaces}

Let $\Sigma$ be a closed surface of \textit{genus $\mathsf{g}$}
with a Riemannian metric $g_{\mu\nu}$. Denote the corresponding volume
element by $\mathrm{d}V\left(z\right)=\sqrt{g(z)}\mathrm{d}^{2}z=\frac{i}{2}\sqrt{g}\mathrm{d}z\wedge\mathrm{d}\bar{z}$
and the total area of $\Sigma$ by $V$. Let $G_{w}\left(z\right)=G\left(z,w\right)$
be the Green function of the (positive) Laplace-Beltrami operator
$-\Delta$ defined \cite{Aubin82} by
\begin{equation}
\begin{cases}
-\Delta G_{w}\left(z\right)=\delta_{w}\left(z\right)-\frac{1}{V},\\
\intop_{\Sigma}G_{w}\mathrm{d}V=0.
\end{cases}
\end{equation}
We also write $\left(-\Delta\right)^{-1}f\left(z\right)=\int G\left(z,w\right)f(w)\mathrm{d}V\left(w\right)$.
It is known that $G\left(z,w\right)=G\left(w,z\right)$ and $G\left(z,w\right)=-\frac{1}{2\pi}\ln d\left(z,w\right)+\mathcal{O}(1)$
as $z\to w$, where $d\left(z,w\right)$ is the geodesic distance.
The \textit{Robin function} of the surface is defined by a regularization
of the Green function:
\begin{equation}
G^{R}\left(x\right)=\lim_{y\to x}\left(G\left(x,y\right)+\frac{1}{2\pi}\ln d\left(x,y\right)\right).\label{eq: robin def}
\end{equation}
For the main properties of Robin functions and their applications see \cite{Flucher99}.
%We also detail the properties of the Robin function on higher genus surfaces in Appendix \ref{sec: Appendix g>0}.

\subsection{Euler equation on surfaces}

An ideal incompressible flow on $\Sigma$ is described by a divergence-free
velocity field $u^{\mu}$, $\nabla_{\mu}u^{\mu}=0$. Such a field admits a unique $L^2$-orthogonal decomposition of the form 
\[
u_\mu =\nabla^\ast_\mu\psi+\text{curl-free}
\]
(we use asterisks to denote the clockwise 90-degree rotation: $a^{\ast\mu}\coloneqq \tensor{\epsilon}{^\mu_\nu}a^\nu$),
where $\psi$ is the \textit{stream function}. The curl-free (and harmonic) part of the flow is present only on multiply connected surfaces, and we will suppress it until Sec.\ \ref{sec: nonzero genus}.

The vorticity of the flow is $\omega=\epsilon^{\mu\nu}\nabla_{\mu}u_{\nu}=-\Delta\psi$. On a closed surface, any incompressible flow has zero total vorticity, $\intop_{\Sigma}\omega\dd V=0$.

We shall also allow the fluid to be electrically charged and placed in a constant uniform magnetic field $B$ orthogonal to the surface. Then for a \textit{smooth flow}, the total energy is
simply the Dirichlet functional of $\psi$, constituting a minimal
coupling of the Euler hydrodynamics to a metric: 
\begin{equation}
H=\frac{1}{2}\int_\Sigma u^2\dd V=\frac{1}{2}\int_\Sigma \lvert\nabla\psi\rvert^2\dd V,\label{eq: Dirichlet energy}
\end{equation}
and the dynamics is governed by the Euler equation
\begin{equation}
\dot{u}^{\mu}+u^{\nu}\nabla_{\nu}u^{\mu}=Bu^{\ast\mu}-\nabla^{\mu}p,
\end{equation}
which is also conveniently written as 
\begin{equation}
\dot{u}^{\mu}-(\omega+B) u^{\ast\mu}+\nabla^\mu\left(\frac{u^2}{2}+p\right)=0.\label{omega-Euler}
\end{equation}
On a sphere (and all simply connected domains) the set of solutions is independent of $B$. The Euler equation for incompressible fluids implies the Helhmholtz vorticity equation:
\begin{equation}
\dot{\omega}+u^{\mu}\nabla_{\mu}\omega=0,\quad u^\mu=\nabla^{\ast\mu}(-\Delta)^{-1}\omega+\text{curl-free}.
\end{equation}
%On a sphere (and all simply connected domains) this equation is equivalent to the Euler equation since $\wt{u}^\mu=0$, and the magnetic field only modifies the pressure. 
In general, the Helmholtz equation does not imply the Euler equation and must be accompanied by the equations of motion of the curl-free harmonic part of $u^\mu$, which we discuss later in Sec.\ \ref{sec: transport}.

\subsection{Vortices on closed surfaces}

A point vortex is usually defined as a flow whose vorticity is localized at one point. However, on a closed surface we must have $\int \omega\,\dd V=0$. This means that a consistent definition of a point vortex must include a compensating negative vorticity. This background vorticity becomes non-dynamical if we choose it to be a constant, i.e. the vorticity of a vortex of strength $\Gamma$ at point $z$ is $\Gamma (\delta_{z}-1/V)$. Therefore we define the (chiral) flow of $N$ point vortices of equal vorticities $\Gamma$ at points $\left\{ z_{i}\right\} _{i=1}^{N}$
by the equations
\begin{equation}
\omega=\Gamma\left(\rho-\frac{N}{V}\right),\,\rho(z)=\sum_{i=1}^{N}\delta_{z_{i}}\left(z\right),\,\psi(z)=\Gamma\sum_{i=1}^{N}G_{z_i}(z).
\end{equation}
We use symbols $x,y,z$ to denote points of the surface and their local complex coordinates at the same time, hopefully not causing any confusion. 

In the presence of point vortices, the energy (\ref{eq: Dirichlet energy})
becomes infinite, so a new form of the Hamiltonian is needed.
% Moreover, ``minimal coupling'' is no longer an option in this case.
% The source of the problem is clear when $H$ is represented as the
% integral of $\psi\cdot\omega$: in the presence of a point vortex
% at $z_{0}$, $\psi\sim\ln d\left(z,z_{0}\right)$ and $\omega\sim\delta_{z_{0}}(z)$,
% so we cannot multiply the two. Alternatively, in the Euler equation,
% $\nabla_{\nu}u^{\mu}$ or $\nabla_{\nu}\omega$ contain delta-like
% singularities that cannot be multiplied by $u^{\nu}$. In the flat
% case the standard way of dealing with this is to simply drop the infinite
% ``diagonal'' terms in the double sum over vortices, which gives
% Kirchhoff's equations (\ref{eq: Kirchhoff}). On a curved surface,
The most natural, from the geometric point of view, resolution of
this singularity is to regularize the infinite self-energy of a vortex
$\omega\psi\sim\Gamma^2\delta_{z_{0}}\cdot G_{z_{0}}$ as in the definition
of the Robin function. Thus the vortex-dependent part of the Hamiltonian, also known as the \textit{Kirchhoff-Routh path function} \cite{Lin41,Lin41a} of the $N$-vortex flow as described above is postulated as
\begin{equation}
H_N=\frac{\Gamma^{2}}{2}\sum_{k}\sum_{j\neq k}G\left(z_{j},z_{k}\right)+\frac{\Gamma^{2}}{2}\sum_{j}G^{R}\left(z_{j}\right).\label{eq:hamiltonian}
\end{equation}
Physically this choice of regularization is equivalent to replacing each point vortex by a disk of a fixed (geodesic) radius $\varepsilon$ of uniform vorticity inside and
subtracting out the large self-energy of the disk $\propto\ln\varepsilon$
(see the standard derivation in flat space e.g. in \cite[\S7.3]{Saffman92}).
In particular, the area of the core of each vortex must be independent
of its position on the surface.

It is known \cite{BoattKoill08} that regardless of the regularization, the result satisfies Kimura's conjecture \cite{Kimura99} that vortex dipoles move along geodesics. Some of the effects
of the Robin function on the dynamics of vortices were explored in
\cites{VitelNels04}{TurVitNel10}. It has also been employed in the
formulation of the quantum Hall effect on surfaces \cite{LasCanWie15}.

% Finally, the Hamiltonian has the same form on surfaces with boundaries,
% as long as $G$ is the Green function with the physically appropriate
% boundary conditions. However, the properties of the Robin function
% change significantly in the presence of boundaries or cycles (for
% instance, on a disk, instead of (\ref{eq: laplacian of robin}) it
% satisfies a Liouville-type equation $-\Delta G^{R}=\frac{R}{4\pi}+\frac{2}{\pi}e^{-4\pi G^{R}}$).

\subsection{Geometric forces}

We assume that all vorticity fluxes are quantized in units of $\Gamma$, which is the case in e.g.\ superfluids. Apart from the optional magnetic field, this is the only energy scale in the problem. Using this scale, we introduce an additional Lorentz-like force into the Euler equation (\ref{omega-Euler}) which is the simplest non-minimal coupling of the Euler equation to a metric and will be important to us in the discussion of odd viscosity and other geometric properties of vortices:
\begin{equation}
\dot{u}^{\mu}-\omega u^{\ast\mu}+\nabla^\mu\left(\frac{u^2}{2}+p\right)=\left(B+\frac{m\Gamma}{4\pi}\left(R-\wb{R}\right)\right) u^{\ast\mu},\label{omega-Euler-R}
\end{equation}
where by the Gauss-Bonnet theorem the average value $\wb{R}$ of the scalar curvature $R$ (twice the Gaussian curvature) of the surface equals
\[
\wb{R}=\frac{1}{V}\int R\dd V=\frac{4\pi\chi}{V}.
\]
The dimensionless parameter $m$, interpreted as the geometric spin of the fluid, takes integer or half-integer values (see Appendix \ref{sec: Appendix gen}). The corresponding Helmholtz equation is
\begin{equation}
\dot{\omega}+u^{\mu}\nabla_{\mu}\left(\omega+\frac{m\Gamma}{4\pi}R\right)=0.\label{Helmholtz-R}
\end{equation}
In other words, the quantity $\omega+\frac{m\Gamma}{4\pi}R$ is convected by the flow instead of the vorticity itself.

In this system the definition of an $N$-vortex flow has to be corrected:
\begin{gather}
    \omega=\epsilon^{\mu\nu}\nabla_\mu u_\nu=-\Delta\psi=-\frac{m\Gamma}{4\pi}\left(R-\wb{R}\right)+\Gamma\left(\rho-\frac{N}{V}\right),\\
    \rho=\sum_{i=1}^N \delta_{z_i},\quad \psi=\Gamma\sum_{i=1}^{N}G_{z_{i}}-m\Gamma U,\label{N-vortex def with R}
\end{gather}
where 
\[
U=(-\Delta)^{-1}\frac{R}{4\pi}\label{eq: def of curvature potential}
\] is the curvature potential.

The corrected vortex Hamiltonian is
\begin{equation}
H_N=\frac{\Gamma^{2}}{2}\sum_{j\neq k}G\left(z_{j},z_{k}\right)+\frac{\Gamma^{2}}{2}\sum_{j}\left(G^{R}\left(z_{j}\right)-2m U\left(z_{j}\right)\right).\label{eq: discrete Hamiltonain (m)}
\end{equation}
The curvature potential happens to coincide with the Robin function up to a constant on surfaces of genus $0$, so in that case we have 
\begin{equation}
H_N^{\sf{g}=0}=\frac{\Gamma^{2}}{2}\sum_{j\neq k}G\left(z_{j},z_{k}\right)+\frac{\Gamma^{2}}{2}(1-2m)\sum_{j}U\left(z_{j}\right).\label{eq: Hamiltonian (m) on sphere}
\end{equation}
A single vortex then moves with velocity $v$ along a level line of $U$ enclosing an area $D$ in such a way that $\oint v^2\dd t\propto (2m-1)\int_D (R-\wb{R})\dd V$. Thus we can call $s=m-1/2$ the \emph{spin of a vortex}. On higher genus surfaces, the Robin function differs from a curvature potential only by an extra term dependent only on the moduli of the underlying Riemann surface, hence this analogy is still meaningful.

%Aside from the explicit introduction of a geometric force into the Euler equation as we have done here, this Hamiltonian can be obtained by modifying the behavior of the core of a point vortex, e.g. by making it an active rotor that responds to the metric, see Appendix \ref{sec: Appendix g=0}.
The Kirchhoff equations express the velocities $v^\mu_i$ of the vortices in terms of this Hamiltonian:
\[v_{i}^{\mu}=\Gamma^{-1}\nabla^{\ast\mu}_{z_{i}}H_N.\label{Kirchhoff covariant}\]

Again, the harmonic part of $u^\mu$ evaluated at $z_i$ has to be included on the r.h.s., but for now we ignore it.

% Here is a way to obtain this model. For simplicity, we work on a deformed
% infinite plane where $G^{R}=\frac{1}{4\pi}\ln\sqrt{g}$. Let us once
% again replace vortices by disks of small radius $\varepsilon$ and
% constant vorticity and take the limit $\varepsilon\to0$, but now
% make the radius dependent on the position of the vortex so that $\varepsilon\cdot\left(\sqrt{g}\right)^{(m-1)/2}$
% is fixed. If $\varepsilon=\frac{\epsilon}{\sqrt{g}^{(m-1)/2}}$ with
% $\epsilon$ constant, then we can easily compute the energy of such
% a disk of vorticity $\omega=\frac{\Gamma}{\pi\varepsilon^{2}}$:
% \begin{equation}
% E_{\mathrm{core}}=-\frac{\Gamma^{2}}{4\pi}\ln\frac{\epsilon}{\sqrt{g}^{(m-1)/2}}+\frac{\Gamma^{2}}{2}G^{R}=-\frac{\Gamma^{2}}{4\pi}\ln\epsilon+\frac{\Gamma^{2}m}{2}G^{R}.
% \end{equation}
% Now, since the first term is completely non-dynamical, we can drop
% it and get the same regularized self-energy. Therefore $m$ can be
% seen as parametrizing the ``dimension'' of the vortex core. On closed surfaces other than tori, $G^{R}$ differs
% from $\frac{1}{4\pi}\ln\sqrt{g}$, and now $\varepsilon\cdot e^{-2\pi(1-m) G^{R}}$ has to be fixed in the limit of small cores. 
% % In fact, $e^{-4\pi G^{R}}\mathrm{d}z\mathrm{d}\bar{z}$,
% % which is $\left(-\Delta G^{R}\right)\mathrm{d}z\mathrm{d}\bar{z}$
% % on flat domains, has long been known to mathematicians as an important
% % natural metric induced by the Robin function \cites{LevenYamag91}{MaitaYamag04}.

\section{Results}

\subsection{Stress tensor}

This is the central section of the present work. Here we introduce
new objects necessary for the coarse-graining procedure and execute
the procedure itself.

The coarse-grained limit is the limit when $N\to\infty$ and the vortex
density $\rho$ can be replaced by a regular positive function with
$\int\rho\mathrm{d}V=N$. To be able to coarse-grain an object like
the energy or the velocity of vortices, we need to obtain expressions
for them that do not refer to individual vortices, and instead involve
only $\rho$ and its functionals. Since the explicit Hamiltonian and
the Kirchhoff equations include summations over $i\neq j$, there
is no way to coarse-grain these objects directly. This is why we first
introduce a new, less singular object which we call the vortex stress tensor.

Denoting by $\rho=\sum_{i}\delta_{z_{i}}$ the
number density of vortices, we define the number current of vortices
as
\begin{equation}
\Gamma\rho v^{\mu}=\sum_{i}\delta_{z_{i}}v_{i}^{\mu}.
\end{equation}
Now, if $H_N$ is understood as a functional of the metric $g_{\mu\nu}$,
we can employ a very general exact Ward-type identity (\ref{eq: Ward})
to rewrite 
\begin{equation}
-\Gamma\rho v^{\ast}_\nu=-\Gamma\sum_{i}\delta_{z_{i}}\epsilon_{\nu\mu} v_{i}^{\mu}=\sum_{i}\delta_{z_{i}}\nabla_{z_{i},\nu}H_N=2\nabla^{\mu}\frac{\delta H_N}{\delta g^{\mu\nu}}.\label{eq: ward for H}
\end{equation}
This motivates us to introduce a new object called the \textit{vortex
stress tensor}, defined as
\begin{equation}
P_{\mu\nu}(z)=2\frac{\delta H_N}{\delta g^{\mu\nu}(z)}\label{eq: P definition}
\end{equation}
(variation at fixed $\left\{z_{i}\right\}$ and $\Gamma$). We will see that $P_{\mu\nu}$ can be coarse-grained trivially, thus becoming a well-defined functional of $\rho$. Furthermore, if the
energy is to be successfully coarse-grained as a functional of $\rho$, denoted $H_{\rm CG}[\rho]$,
it must satisfy the continuum analog of the last equality in (\ref{eq: ward for H}):
\begin{equation}
\rho\nabla_{\mu}\frac{\delta H_{CG}}{\delta\rho}=\nabla^{\nu}P_{\nu\mu}\label{eq: continuous Ward}
\end{equation}
(see Appendix \ref{sec: Appendix g=0} for a proof). Now (\ref{eq: continuous Ward}) immediately gives the coarse-grained
Hamiltonian $H_{\rm CG}$ and (\ref{eq: ward for H}) gives the velocity vector
field of the coarse-grained flow of vortices:
\begin{equation}
v^{\mu}=\nabla^\ast_{\nu}\frac{\delta H_{\rm CG}}{\delta\rho}. \label{coarse-grained v in terms of H}
\end{equation}

Now we restate the last few relations in a compact form and give an
explicit expression for $P_{\mu\nu}$ as a functional of $\rho$. 
%In Appendix \ref{sec: Appendix g>0} we point out the difficulties with computing $P_{\mu\nu}$ when $\mathsf{g}\geq2$.
% In Appendix \ref{sec: Appendix g=1} we also derive the stress tensor on a torus. The absence of a fully explicit formula for $\mathsf{g}>1$ does not, however, prevent us from coarse-graining the flow.

\begin{prop}
    \label{prop: Stress}For any number $N$ of vortices of the same vorticity
    $\Gamma$ with the Hamiltonian (\ref{eq: discrete Hamiltonain (m)}),
    located at points $\left\{ z_{i}\right\} _{i=1}^{N}$ of a closed
    surface of genus $\mathsf{g}$ with a Riemannian metric $g_{\mu\nu}$,
    their equations of motion (Kirchhoff equations) (\ref{Kirchhoff covariant}) can be written as
    \begin{equation}
        \Gamma\rho v^\ast_\mu=-\nabla^{\nu}P_{\nu\mu},\label{eq: divergence of P}
    \end{equation}
    where $\rho=\sum_{i}\delta_{z_{i}}$ is the number density of vortices,
    $\rho v^{\mu}=\sum_{i}\delta_{z_{i}}v_{i}^{\mu}$ is the number
    current of vortices and $P_{\mu\nu}$ is the symmetric ``vortex stress''
    tensor defined in (\ref{eq: P definition}). In this equation we have ignored the possible global circulations along the cycles of the surface, which we elaborate on in Sec.\ \ref{sec: nonzero genus}. The stress tensor can be decomposed as
    \begin{equation}
        P_{\mu\nu}=P^{C}_{\mu\nu}+P^A_{\mu\nu}+2mP^S_{\mu\nu}+o_{\mu\nu},\label{eq: stress P}
    \end{equation}
    where the last term is an irrelevant divergenceless tensor (moreover it is flow-independent, i.e.\ independent of the $z_i$'s, at least on surfaces of genus 0 and 1). The first term has the form of the naive ``classical'' stress expected for a continuous vorticity distribution,
    \[
        P^C_{\mu\nu}=\left[-\nabla_\mu\psi\nabla_\nu\psi+\frac{1}{2}g_{\mu\nu}\lvert\nabla\psi\rvert^2\right]+g_{\mu\nu}\frac{(N-m\chi)\Gamma}{V}\psi.
    \]
    $P^S_{\mu\nu}$ is the stress corresponding to the Lorentz-like spin force
    \[
        P^S_{\mu\nu}=\tau_{\mu\nu}+g_{\mu\nu}\frac{\Gamma^{2}}{8\pi}\rho,
    \]
    where 
    \begin{equation}
        \tau_{\mu\nu}=\frac{\Gamma}{4\pi}\left[\nabla_{\mu}\nabla_\nu\psi-\frac{1}{2}g_{\mu\nu}\Delta\psi\right]-g_{\mu\nu}\frac{m\Gamma^2}{32\pi^2}R.\label{eq: stress tensor T}
    \end{equation}
    $P^S_{\mu\nu}$ is divergenceless wherever curvature is zero. Finally, $P^A_{\mu\nu}$ is the remaining \emph{anomalous stress}
    \[
        P^A_{\mu\nu}=-g_{\mu\nu}\frac{\Gamma^{2}}{8\pi}\rho.
    \]
    In these formulas, $\psi$ is the stream function (\ref{N-vortex def with R}) for the $N$-vortex flow.
    
    The second order pole type singularities at $z_i$ in $P^C_{\mu\nu}$ need to be understood in a specific distributional sense, see Appendix \ref{sec: Appendix g=0}.

    For additional remarks on the case $\mathsf{g}\geq1$ see Sec.\ \ref{sec: nonzero genus} and Appendix \ref{sec: Appendix g=1}.
\end{prop}

We call $P^A_{\mu\nu}$ the anomalous stress because it is associated with the ``excluded volume'' of vortex cores, enforced by the summation rule $j\neq k$ in the pair interaction (\ref{eq: discrete Hamiltonain (m)}). More concretely, it comes from the self-energy of the vortex cores represented by the Robin function. Unlike $P^C_{\mu\nu}$, which is of order $N^2$, the anomalous term is of order $N$.

Nevertheless, the notion of the anomalous stress is ambiguous, as $P^A_{\mu\nu}$ is not the only contribution of the Robin function to the stress. This ambiguity is related to the discrepancy between the spins of the original fluid and of the vortices in it. Namely, using the identity $P^A_{\mu\nu}=\tau_{\mu\nu}-P^S_{\mu\nu}$, we can rewrite (\ref{eq: stress P}) as
\[
   \mathsf{g}=0:\quad P_{\mu\nu}=P^C_{\mu\nu}+\tau_{\mu\nu}+(2m-1)P^S_{\mu\nu}+o_{\mu\nu}.\label{eq: alternative anomalous stress}
\]
We recognize the third term as the spin of vortices, corresponding to the last term in the vortex Hamiltonian (\ref{eq: Hamiltonian (m) on sphere}). We can therefore interpret $\tau_{\mu\nu}$ as the alternative anomalous stress, first introduced as the traceless part of the anomalous momentum flux tensor in \cite{WiegmAban14}. While $P^A_{\mu\nu}$ is the stress that appears anomalous from the point of view of the Euler equation, $\tau_{\mu\nu}$ is the stress anomalous to the Kirchhoff equations. It is $\tau_{\mu\nu}$ that is sometimes called \emph{odd stress}.

\subsection{Coarse-grained vortex flow}
We have expressed $P_{\mu\nu}$ solely in terms of the vortex density $\rho$. The possibility of doing so is a remarkable property of the logarithmic pair interaction of vortices. This form of the tensor allows for an easy passage to the coarse-grained limit: it is natural to assume that the formula for $P_{\mu\nu}$ in terms of $\rho$ remains unchanged after coarse-graining. In particular, the formula for the stream function always remains $\psi=\Gamma(-\Delta)^{-1}\rho-m\Gamma U$.

With this assumption, formulas (\ref{eq: continuous Ward}-\ref{coarse-grained v in terms of H}) let us find the anomalous difference between the vortex flow $v^\mu$ and the flow of the original fluid $u^\mu$.

\begin{prop}[Coarse-grained flow]
    \label{prop: Coarse-grained flow} In the limit $N\to\infty$ on a closed surface of any genus, if the vortex density $\rho=\sum_{i=1}^{N}\delta_{z_{i}}$ approximates a continuous (``coarse-grained'') distribution $\rho$ with $\int\rho\mathrm{d}V=N$, the coarse-grained vortex flow $v^{\mu}$ is incompressible and deviates from the ``naive'' flow $u^{\mu}$ with the continuous stream function $\psi=\Gamma(-\Delta)^{-1}\rho-m\Gamma U$ as follows:
    \begin{equation}
        v^{\mu}=u^{\mu}-\frac{\Gamma}{8\pi}\nabla^{\ast\mu}\ln\rho.\label{eq: v in terms of u (m)}
    \end{equation}
    Stationary solutions ($v^{\mu}=\dot{u}^{\mu}\equiv0$) with nowhere vanishing $\rho$ must satisfy $\psi=\frac{\Gamma}{8\pi}\ln\rho+\mathrm{const}$, equivalent to the Liouville-type ``mean field'' equation on $\rho$:
    \begin{equation}
        \rho+\frac{1}{8\pi}\Delta\ln\rho=\frac{m}{4\pi}R+\frac{k}{V},\quad k=N-m\chi.\label{eq: Liouville-type equation}
    \end{equation}
    In the large-$N$ limit it has a solution that is close to uniform density and has the following gradient expansion in terms of curvature:
    \[\rho=\frac{k}{V}+\frac{m}{4\pi}\left(R-\frac{V}{k}\frac{1}{8\pi}\Delta R+\mathcal{O}(k^{-2})\right).\]
\end{prop}

We emphasize the disappearance of the Robin function from the coarse-grained dynamics of vortices. In particular, in the $m=0$ model the two-particle contribution combines with the single-particle effect of the Robin function (the latter was discussed e.g.\ in \cite{TurVitNel10}) leading to the full cancellation of the curvature response in the coarse-grained dynamics.

The above result therefore derives from a microscopical description
and generalizes the mean field equations of
%\cite{CagLiMaPu92,CagLiMaPu95,Kiessling93}.
\cites{MontgJoyce74}{Kiessling93}.
%The solvability of equations of the form (\ref{eq: Liouville-type equation}) was also examined recently in \cite{BarGuJeMo18}. 
One obvious stationary solution for $m=0$ is $\rho=\frac{N}{V}=\mathrm{const}$,
but for a generic $N$ this equation may have many solutions. Since $N$ is large, non-uniqueness e.g.\ on a flat torus follows from the results in 
%\cite[Thm.\ 1.3]{Lin00}\cite{LinLucia06}
\cite{RicciTaran98}.

\begin{prop}
The coarse-grained Hamiltonian
on any genus equals (up to flow-independent constants and ignoring the harmonic part of $u^\mu$)
\begin{multline}
H_{\rm CG} [\rho]=\frac12\int\left[u^{2}-\frac{\Gamma^{2}}{4\pi}\rho\ln\rho\right]\dd V=\\
\\=\frac{\Gamma^2}{2}\int\left[\rho\cdot(-\Delta)^{-1}\rho-\frac{1}{4\pi}\rho\ln\rho-2m R\cdot(-\Delta)^{-1}\rho\right]\dd V,\label{eq: coarse-grained H(m)}
\end{multline}
where $u^\mu=\nabla^{\ast\mu}\left((-\Delta)^{-1}\rho-m\Gamma U\right)$ is the coarse-grained flow of the original fluid in terms of the coarse-grained density $\rho$ of vortices.
The $\rho$-$\rho$ Poisson brackets are given by 
\begin{align}
\left\{ F_1,F_2\right\} \left[\rho\right]&=\Gamma^{-1}\int\rho\cdot\epsilon^{\mu\nu}\nabla_{\mu}\frac{\delta F_1}{\delta\rho}\nabla_{\nu}\frac{\delta F_2}{\delta\rho}\dd V\label{Poisson brackets new}
\end{align}
for any two functionals $F_1$ and $F_2$ of $\rho$. The brackets with the harmonic part of the flow will be stated in Section \ref{sec: transport}, and the resulting equation of motion for $\rho$ is, as expected, $\dot\rho=\{\rho,H_{\rm CG}\}=-v^\mu\nabla_\mu\rho$.
\end{prop}
The term $\int \rho\ln\rho\dd V$, like any local function of $\rho$, is a Casimir in these Poisson brackets, but it does affect the energy and the vortex flow.

\subsection{Discussion}

\subsubsection{Vortex matter and negative temperatures}

The coarse-grained energy that we have obtained directly from the Hamiltonian of point vortices contains a term that looks like the von Neumann entropy. However, this is not an entropy since we have done everything at zero temperature (unlike in the statistical mechanics approaches such as \cite{LundgPoin77}).
%Unlike previous approaches that were based on the statistical mechanics of vortices via the Boltzmann entropy functional \cite{LundgPoin77,RoberSomm91,Kiessling93,KiessWang12}, we have derived via coarse-graining the ``regularized'' energy functional directly from the Hamiltonian of point vortices. 
Interestingly, the result is identical to the free energy of a statistical vortex ensemble at a special negative temperature.

Namely, if $H$ in (\ref{eq: coarse-grained H(m)})
is interpreted as a free energy $H=\Gamma^{2}\left(E-S/\wt{\beta}\right)$
with entropy $S=-\int\rho\ln\rho\,\dd V$, this system appears to
have precisely the negative inverse temperature $\wt{\beta}=-8\pi$. There is only a limited understanding of the equilibria at $\wt{\beta}\leq-8\pi$ in some special cases \cite{CagLiMaPu92,CagLiMaPu95}.

In the context of the Dyson gas \cite[(1.1), (1.4), (2.11)]{ZabroWieg06},
the ratio between the corresponding terms in the free energy is $\wt{\beta}=8\pi\frac{2-\beta}{\beta}$, where $2\beta$ is the power of the Vandermonde determinant in the corresponding ensemble. There, $\beta=\frac{\Gamma^{2}}{k_{B}T}$ is interpreted
as the inverse temperature of the gas, which is believed to form a
Wigner crystal at large $\beta$, i.e.\ precisely when $\wt{\beta}\searrow-8\pi$.

\subsubsection{Odd viscosity\label{subsec:Odd-viscosity-and}}

On flat regions of the surface, the $m$-dependent part $2mP^S_{\mu\nu}$ of the stress tensor is divergenceless, making it impossible to observe the real value of $m$ in the bulk of the vortex matter.

The traceless part of $\tau_{\mu\nu}$ (\ref{eq: stress tensor T}) is also
known as \textit{odd stress} and in complex coordinates ($4\tau_{zz}=\tau_{xx}-\tau_{yy}-i\left(\tau_{xy}+\tau_{yx}\right)$,
$2u_{z}=u_{x}-iu_{y}=2i\partial_{z}\psi$) it is $\tau_{zz}=-2\eta i\partial_{z}u_{z}=2\eta\partial_{z}^{2}\psi$
with the odd viscosity coefficient
\begin{equation}
\eta=\frac{\Gamma}{8\pi}.
\end{equation}
In \cite{WiegmAban14}, only the case $m=1/2$ was examined, when this is the only surviving new term in the stress, see (\ref{eq: alternative anomalous stress}). Note
that $\eta$ matches the coefficient of the anomalous term
$\rho\ln\rho$, since both originate from the same phenomenon.

\subsubsection{Momentum flux tensor}

$P_{\mu\nu}$ is not a conserved current in the sense that its divergence does not equal the time derivative of any quantity that has appeared so far. However, using the continuity equation $\dot\rho+v^\mu\nabla_\mu\rho=0$, we
observe the relations
\begin{equation}
-\epsilon^{\mu\nu}\nabla_{\mu}\nabla_{\lambda}P_{\nu\lambda}=\nabla\times\left(\Gamma\rho\bm{v}^{\ast}\right)=\Gamma\dot{\rho}=\nabla\times\dot{\bm{u}}.
\end{equation}
This lets us find $\dot{u}^{\mu}$ up to a gradient term (assuming no external forces aside from the ones we have already announced). This way we recover the original Euler equation that, up to a redefinition of $p$, can be now written as
\begin{equation}
\dot{u}^{\mu}+\nabla_{\nu}P^{\nu\mu}=-\nabla^{\mu}\left(\frac{u^{2}}{2}+p\right).\label{eq: anomalous euler with P}
\end{equation}
The pressure $p$ can be found from the incompressibility constraint $\nabla_{\mu}u^{\mu}=0$. Therefore $P_{\mu\nu}$ differs from the momentum flux tensor $\Pi_{\mu\nu}$ of this anomalous Euler equation only by a trace term:
\begin{gather}
\dot{u}^{\mu}+\nabla_{\nu}\Pi^{\nu\mu}=0,\\
\Pi_{\mu\nu}=P_{\mu\nu}+\frac{1}{2}g_{\mu\nu}\left(u^{2}+2p\right).\label{eq: momentum flux tensor}
\end{gather}

Whereas the momentum flux tensor provides the \emph{acceleration} of the fluid, $P_{\mu\nu}$ describes the \emph{velocities} of vortices. This observation has to do with the fact that the vortex Hamiltonian and symplectic form can be obtained by a reduction of the original infinite-dimensional system with respect to the relabeling symmetry \cite{MarsdWein83}, giving equations of motion that are first order in time.

Moreover, $P_{\mu\nu}$ does not account for the global circulations encoded in the harmonic part of $u^\mu$ (see Sec.\ \ref{sec: nonzero genus}). Therefore it can be understood as exactly the part of the momentum flux tensor that generates local torques (shears) in the original fluid, driven by the vortices. All other forces acting on the fluid are curl-free, don't create any local torques and therefore don't influence the motion of vortices.

\subsubsection{Details on nonzero genus}\label{sec: nonzero genus}

Here we list the necessary changes to the above formulas if one restores the harmonic part of the flow $u^\mu$ (present only for $\mathsf{g}\geq 1$), which we have ignored so far.

According to Hodge theory, an incompressible flow admits a unique $L^2$-orthogonal decomposition of the form 
\[
u_\mu =\wh u_\mu+\wt u_\mu=\nabla^\ast_\mu\psi+\wt u_\mu,\quad \epsilon^{\mu\nu}\nabla_\mu \wt u_\nu=0.
\]
The total energy of such a flow now includes the energy of $\wt{u}^\mu$:
\[H=\wh{H}+\wt{H},\quad \wh{H}=\frac{1}{2}\int_\Sigma \lvert\nabla\psi\rvert^2\dd V,\quad \wt{H}=\frac{1}{2}\int_\Sigma \wt{u}^2\dd V.\]
The Euler equation describes the dynamics of both $\wh{u}^\mu$ and $\wt{u}^\mu$. The vortex Hamiltonian (\ref{eq: discrete Hamiltonain (m)}) should really be denoted $\wh{H}_N$, and the velocities of vortices need to include the value of $\wt{u}^\mu$:
\[v_{i}^{\mu}=\wh v_{i}^{\mu}+\wt u(z_i),\quad \wh{v}_i^\mu=\Gamma^{-1}\nabla^{\ast\mu}_{z_{i}}\wh{H}_N.\]
Therefore the curl-free part of the flow $\wt{u}^\mu$ has to be included in the symplectic structure for vortices, correcting the symplectic form proposed in \cite{BoattKoill08}. The full Poisson brackets are stated in Sec.\ \ref{sec: transport}.

The Ward-type identity (\ref{eq: ward for H}) gets corrected as
\begin{equation}
-\Gamma\rho \wh v^{\ast}_\nu=2\nabla^{\mu}\frac{\delta \wh H_N}{\delta g^{\mu\nu}}.
\end{equation}
The stress $P_{\mu\nu}$ can be defined as the variation of the total energy $H_N=\wh H_N+\wt{H}$ taken at fixed $\wt{H}$, in addition to previous constraints. The statement of Proposition \ref{prop: Stress} is now
\begin{equation}
        \Gamma\rho v^\ast_\mu=-\nabla^{\nu}P_{\nu\mu}+\Gamma\rho \wt u^\ast_\mu,
\end{equation}
and the Euler equation in terms of $P_{\mu\nu}$ reads
\begin{equation}
\dot{u}^{\mu}+\nabla_{\nu}P^{\nu\mu}=\left(\omega+B+\frac{m\Gamma}{4\pi}(R-\wb{R})\right)\wt{u}^{\ast\mu}-\nabla^{\mu}\left(\frac{u^{2}}{2}+p\right).
\end{equation}
All terms except $B\wt{u}^\mu$ on the r.h.s.\ can be written as the divergence of a symmetric tensor, therefore we can redefine $\Pi_{\mu\nu}$ as
\begin{equation}
\dot{u}^{\mu}+\nabla_{\nu}\Pi^{\nu\mu}=B\wt{u}^{\ast\mu},
\end{equation}
and the expression for $\Pi_{\mu\nu}$ is the one for $P_{\mu\nu}$ after the substitution $\nabla_\mu\psi\mapsto-u^\ast_\mu$ in the traceless part and the replacement of the trace by $u^2/2+p$ as before.

\section{Conclusions}

We have developed a new framework for point vortices on curved surfaces and coarse-grained the flows of vortices on arbitrary closed surfaces. Our method of coarse-graining should work in principle on surfaces with boundaries, although we expect many nontrivial effects related to the appearance of boundaries of the vortex matter itself, not coinciding with the boundaries of the surface \cite{BogatWiegm19}.

In addition, we have related the anomalous behavior of the vortex matter to its odd viscosity, which is found to be a universal fraction of the vorticity quantum $\Gamma$. We have distinguished this effect from the completely classical effect of spin, which is merely a curvature-dependent force acting on the fluid. It is apparent that odd viscosity is related to the anomalous difference of $1/2$ between the spins of the original fluid and of discrete vortices in it. Comparing the vortex Hamiltonian with the free energy of a Laughlin state of particles with conformal spin $s$ \cite{FerraKlev14}, we find 
\begin{equation}
2m=1+2s.
\end{equation}
This further justifies treating $s=m-1/2$ as the spin of a vortex.

Finally, the anomalous negative ``temperature'' of vortex matter might explain the recent observations of negative-temperature states of vortices in the Bose-Einstein condensate \cites{GaRYBBBRDN18}{JoGrStBiSH18}.
% Onsager famously argued \cite{Onsager49} that negative-temperature
% states of vortices in fluids of finite volume would exhibit clusterization into two large clusters
% organized by the sign of the vorticity (see \cite[Ch. 6]{Newton01}
% for an extended discussion). This so-called \textit{Onsager clusterization
% }was long elusive in experiments but has finally been observed recently
% in the Bose-Einstein condensate \cites{GaRYBBBRDN18}{JoGrStBiSH18},
% confirming the effective negative shift of the inverse temperature
% scale. The temperature $\wt{\beta}=-8\pi$ is also known as the
% \textit{supercondensation }or\textit{ concentration} \textit{transition}
% when the large vortex clusters collapse into points, at least in some
% geometries \cite{YuBiNiReBr16}. This temperature plays the role of
% the ``effective absolute zero'' for this system.

\subsection*{Acknowledgements}

I am grateful to Alexander Abanov, William Irvine, Semyon Klevtsov,  Alexios Polychronakos, and Vincenzo Vitelli for helpful discussions, and to Paul Wiegmann for guidance during the preparation of this work.  
\bigbreak
\PRLsep

\setlength\bibhang{0pt} % remove left margin in citations
\setlength\bibitemsep{0pt} %set empty space between bib items
%\raggedright %release right end of bib entries
%\vspace*{-3em}
\setstretch{0.85} %reduce line spacing
\printbibliography[heading=bibintoc]
\setstretch{1} %restore line spacing

\appendix
\renewcommand{\theequation}{\thesection.\arabic{equation}}
\setcounter{equation}{0}

\section{Applications}

\subsection{Quantization of spin}

First we argue that the spin parameter $m$ in the Euler equation (\ref{omega-Euler-R}) has to be integer or half-integer, provided that all vorticity is quantized in units of $\Gamma$. The argument is analogous to the Dirac quantization condition. Namely, if we want (\ref{omega-Euler-R}) to apply to arbitrary initial vorticity distributions, then we can consider a solution with a single point vortex (an analog of a Dirac string)
\[\omega=-\frac{m\Gamma}{4\pi}R+A\delta_{z_0}.\]
Because of the neutrality condition, $A=m\chi\Gamma$. Finally, the quantization of vorticity requires that $m\chi\in\mathbb{Z}$. Since $\chi$ is even, this means that most generally $m\in\mathbb{Z}/2$.

%An alternative argument can be given based on an attempt to quantize the fluid. In the presence of a magnetic field and curvature, the canonical momentum of the fluid is the one-form
% \[
% J=u+A+\frac{m\Gamma}{4\pi}A_R,
% \]
% where $A$ is the vector potential of the magnetic field, $\dd A=B\dd V$, and $A_R$ is the spin connection, $\dd A_R=R\dd V$. Then the curvature of the connection specified by $J$ equals
% \[\varpi=\dd J=\left(\omega+B+\frac{m\Gamma}{4\pi}R\right)\dd V.\]
% Integrating both sides over the entire surface, we find
% \[c_1\coloneqq\Gamma^{-1}\int_\Sigma\varpi=N_\Phi+m\chi,\quad N_\Phi=\Gamma^{-1}\int B\dd V.\]
% If we were to require the quantization of all fluxes of canonical momentum, we would have another argument for the quantization of spin $m\in\mathbb{Z}/2$. This is equivalent to requiring that the wave function of the fluid $\Psi$ be such that $(d-i\Gamma^{-1}J)\Psi=0$, meaning that $\Psi$ is a horizontal section of a complex line bundle with the connection $i\Gamma^{-1}J$ (with $\Gamma$ being an analog of $\hbar$). Then the momentum $J$ the gradient of the phase of the wave function. Such a line bundle only exists if its curvature integrates to an integer, which is exactly the condition above.

\subsection{Conical singularities}

In our derivation of the Liouville equation (\ref{eq: Liouville-type equation}) we have tacitly assumed smoothness of the coarse-grained density $\rho$, which is why strictly speaking it cannot be directly applied to surfaces with singular curvature. Here we only make a naive attempt to extract the leading behavior of the density. 

Consider a surface whose curvature is concentrated at one point $z_0$:
\[R(z)=4\pi\alpha\delta_{z_0}(z)+\mathcal{O}(1),\quad z\to z_0.\]
Such a singularity corresponds to a conical point with the cone angle $2\pi(1-\alpha)$. We assume $0<\alpha<1$. Let us look for solutions of (\ref{eq: Liouville-type equation}) that turn to zero at $z_0$ as a power-law of the distance $r$ from $z_0$:
\[\rho\sim r^{2\gamma}.\]
We have $\Delta\ln\rho\sim 4\pi\gamma \delta_{z_0}$, which matches the singular term on the right hand side of (\ref{eq: Liouville-type equation}) if and only if
\[\gamma/\alpha=2m\in \mathbb{Z}.\]
%This relation correctly quantifies the classical dimension of the ``vertex operator'' corresponding to the insertion of a conical singularity \cite{CanChiLaWi16}.

Moreover, since away from the singularity we have $\rho\approx \rho_{\infty}=k/V$, we find the number of vortices concentrated at the singularity
\[\int (\rho-\rho_{\infty})\dd V=\frac{m}{4\pi}\int R\dd V=m\alpha.\]
These formulas are similar to those for the classical contributions to the density of Quantum Hall fluids on a cone \cite{CanChiLaWi16}.

\subsection{Transport of vortices}\label{sec: transport}

The curl-free part $\wt{u}^\mu$ of the incompressible flow on $\Sigma$ of genus $\mathsf{g}$ can be decomposed as
\[
\wt u_\mu=\sum_{a=1}^{2\mathsf{g}} \gamma_a\theta^a_\mu,
\]
where $\theta^a_\mu$ are fixed curl-free \textit{and} divergence-free (closed and co-closed) one-forms, and can be chosen so that $\theta^{\mathsf{g}+a}_\mu=\tensor{\epsilon}{_\mu^\nu}\theta^a_\nu$ and $\int g^{\mu\nu} \theta^a_\mu \theta^b_\nu \dd V=\delta^{ab}$. They form a harmonic orthonormal basis of the de Rham cohomology $H^1(\Sigma,\mathbb{R})$. The numbers $\gamma_a=\int_\Sigma u^\mu \theta^a_\mu \dd V$ parametrize the space $H^1(\Sigma,\mathbb{R})$ and express the mean ``homological circulations''. Namely, if we let $C_a$ be a fundamental loop on $\Sigma$ and choose $\theta^a$ so that $\oint_{C_a} \theta^b=c_a\delta^{ab}$, then $\gamma_a c_a$ is the extra non-vorticity-generated circulation included in $\oint u$ evaluated along any loop homological to $C_a$.

The energy of the flow can now be rewritten as
\begin{equation}
H=\wh{H}+\frac{1}{2}\sum_a\gamma_a^2
\end{equation}
and the dynamics of the vorticity $\omega$ and the $\gamma_a$'s is determined by the Poisson brackets
\begin{align}
\left\{ F,G\right\} \left[\omega\right]&=\int\left(\omega+\frac{m\Gamma}{4\pi}R\right)\cdot\epsilon^{\mu\nu}\nabla_{\mu}\frac{\delta F}{\delta\omega}\nabla_{\nu}\frac{\delta G}{\delta\omega}\dd V,\label{Poisson brackets}\\
\left\{F[\omega],\gamma_a\right\}&=\int \left(\omega+\frac{m\Gamma}{4\pi}R\right)\cdot \theta_\mu^a\nabla^\mu \frac{\delta F}{\delta\omega}\dd V,\\
\left\{\gamma_a,\gamma_b\right\}&=\int\left(\omega+B+\frac{m\Gamma}{4\pi}\left(R-\wb{R}\right)\right)\cdot \epsilon^{\mu\nu}\theta^a_\mu\theta^b_\nu \dd V
\end{align}
for any functionals $F$ and $G$ of $\omega$. The complete equations of motion then read $\dot\omega=\{H,\omega\}$, which is (\ref{Helmholtz-R}), and 
\[\dot\gamma_a=\{H,\gamma_a\}=\int \left(\omega+\frac{m\Gamma}{4\pi}R\right)\theta^a_\mu u^{\ast\mu}\dd V+\sum_b\gamma_b\{\gamma_b,\gamma_a\}.\]
We see that the uniform magnetic field affects the dynamics only on multi-connected surfaces through the Poisson brackets between $\gamma_a$'s. The structure of incompressible flows on surfaces of arbitrary genus was studied in \cites{IzosKheMou15}{IzosiKhesi17}.

The tremendous effect that the uniform magnetic field can have on the dynamics of the fluid on multi-connected surfaces can be easily seen in the simplest example of a torus.

Consider a flat torus with a metric corresponding to a rectangle of size $L_1\times L_2$ with identified edges. In the standard Cartesian coordinate system, we can let $\bm{\theta}^a$, $a=1,2$ be the $L^2$-normalized coordinate vector fields. Let the torus be uniformly filled with vortices, $\rho=N/V$. This is clearly a configuration with zero coarse-grained vorticity, which means that $\wh{\bm{v}}=\wh{\bm{u}}=0$. Furthermore, let us enable a constant and uniform external force by adding the term $\bm{f}=f_1 \bm{\theta}^1+f_2\bm{\theta}^2$ to the right-hand side of our Euler equation. Due to translational symmetry, we also have $\bm{\nabla}(u^2/2+p)=0$.

The Euler equation (\ref{omega-Euler}) then takes the form 
\[
\dot{\bm{u}}=B\wt{\bm{u}}^\ast+\bm{f}.
\] 
This equation admits a simple solution -- the steady uniform flow $\bm{u}=B^{-1}\bm{f}^\ast$ \textit{orthogonal to the direction of the force}. More generally, a rotating initial state $\bm{u}=\gamma_1\bm{\theta}^1+\gamma_2\bm{\theta}^2$ moves according to $\dot\gamma_1=B\gamma_2+f_1$, $\dot\gamma_2=-B\gamma_1+f_2$. The apparent Hall conductance $\rho\bm{v}=\sigma^H \bm{f}^\ast$ is found to be equal to $\sigma^H=\frac{\Gamma N}{BV}=\nu$, the ``filling fraction''.

This argument can also be interpreted as a hydrodynamic analog of Laughlin's pumping argument \cite{Laughlin81}, where enabling a temporary electric force $\bm{f}$ that slowly increases the magnetic flux along a non-contractible cycle on the torus by $\Gamma$ leads to $\nu$ vortices being transported in a direction orthogonal to $\bm{f}$. Of course, the force $\bm{f}$ can be of any nature, as long as over the time of its presence it pumps total momentum flux of magnitude $\Gamma$ into the system.

%Lastly, let us consider the case $B=0$, and let vortices be discrete again, so $\rho=\sum_i\delta_{z_i}$. Due to translational symmetry, $\dot\gamma_a=-\Gamma\int \rho  v^{\ast\mu}\theta^a_\mu\dd V=\int \nabla^\nu P_{\nu\mu}\theta^{a\mu}\dd V=0$, and moreover, in the absence of the magnetic field, $\{\gamma_a,\gamma_b\}=0$. This means that $\gamma_a$'s in this case can be treated entirely as external parameters for the system.

\subsection{Rotating surfaces}

In the plane rotating at an angular
frequency $\Omega$, Kirchhoff's equations read
\begin{equation}
i\dot{\bar{z}}_{i}=-\Omega\bar{z}_{i}+\frac{1}{2\pi}\sum_{j\ne i}\frac{\Gamma}{z_{i}-z_{j}}.\label{eq: Kirchhoff-Omega}
\end{equation}
In the Hamiltonian formalism for vortices, this corresponds to an addition of the centrifugal potential $\Gamma\Omega\sum_i |z_i|^2$ to the energy. The analog of this potential on arbitrary surfaces is the addition of a large vortex of strength
$2\Omega V$ and same sign as $\Gamma$ at a fixed point $z_{0}$ (viewed as the ``infinity'') of our closed surface. For simplicity we only consider the case $\mathsf{g}=m=0$
here. The Hamiltonian for this system is
\begin{equation}
H_N\mapsto H_N+2\Gamma\Omega V\sum_{j>0}G_{z_{0}}\left(z_{j}\right),
\end{equation}
and the corrected stress tensor is given by the same formulas with
the substitutions 
\begin{align}
\psi & \mapsto \psi+2\Omega V\cdot G_{z_{0}},\\
\omega & \mapsto \omega -2\Omega+2\Omega V\delta_{z_{0}}.
\end{align}
In the lowest-energy states, vortices are repelled from $z_0$ and it is easy to see that the Liouville equation (\ref{eq: Liouville-type equation}) after a shift of the l.h.s.\ by $2\Omega$ does not have any smooth solutions. Instead, the coarse-grained density is nonzero only inside a bounded domain whose area is approximately equal to $N\Gamma/2\Omega$. One can introduce a notion of angular impulse (or generalized angular momentum) for such configurations, which we do in the next section.

\subsection{Sum rules\label{subsec:Sum-rules}}

Kirchhoff's equations (\ref{eq: Kirchhoff-Omega}) satisfy an exact sum
rule for the angular impulse \cite{Saffman92}:
\begin{equation}
    L=\frac{i}{2}\sum_{i}\left(z_{i}\dot{\bar{z}}_{i}-\dot{z}_{i}\bar{z}_{i}\right)=\frac{\Gamma}{4\pi}N\left(N-1\right)-\Omega\sum_{i}\left|z_{i}\right|^{2}.\label{eq: old sum rule}
\end{equation}
The term linear in $N$ can be called anomalous since it comes from the exclusion of self-interactions of vortices. 
The generalization of the sum rule to closed surfaces requires a choice of an ``axis
of rotation''. For this purpose we once again pick a point $z_{0}\in\Sigma$
and consider the potential $G_{z_{0}}$. The angular impulse relative
to $z_{0}$ is defined as
\begin{equation}
    L=2V\intop_{\Sigma\setminus\left\{ z_{0}\right\} }\rho\nabla^{\mu}G_{z_{0}}\epsilon_{\mu\nu}v^{\nu}\mathrm{d}V=\frac{2V}{\Gamma}\intop_{\Sigma\setminus\left\{ z_{0}\right\} }P_{\mu\nu}\nabla^{\mu}\nabla^{\nu}G_{z_{0}}\dd V.
\end{equation}
This quantity is conserved only if the surface has a special symmetry
w.r.t.\ $z_{0}$, e.g.\ if it is axisymmetric like an ellipsoid. More
precisely, the vector field $\epsilon^{\mu\nu}\nabla_{\nu}G_{z_{0}}$
has to be a (singular) Killing vector.

Moreover, we see immediately that the anomalous stress $P^A_{\mu\nu}$ is directly related to the anomalous angular impulse:
\[L^A=\frac{2V}{\Gamma}\intop_{\Sigma\setminus\left\{ z_{0}\right\} }P^A_{\mu\nu}\nabla^{\mu}\nabla^{\nu}G_{z_{0}}\dd V=-\frac{\Gamma}{4\pi}N,\]
since $\Delta G_{z_0}=1/V$ away from $z_0$.

Finally, using the relation between $v$, $u$ and $\rho$, we obtain the formula for $L$ in terms of the density:
\begin{gather}
L =\frac{\Gamma}{4\pi}\iint\rho(z)L\left(z,w\right)\rho(w)\mathrm{d}V(z)\mathrm{d}V(w)-\frac{\Gamma}{4\pi}N,\\
L\left(z,w\right) =-4\pi V\left(\nabla_{\mu}G_{z_{0}}(z)\nabla_{z}^{\mu}G\left(z,w\right)+z\leftrightarrow w\right).
\end{gather}
If we also enable the
``solid rotation'' of angular frequency $\Omega$ as above, we arrive
at the generalized sum rule
\begin{multline}
L=\frac{\Gamma}{4\pi}\iint\rho(z)L\left(z,w\right)\rho(w)\mathrm{d}V(z)\mathrm{d}V(w)-\frac{\Gamma}{4\pi}N-\\
-\left(4\Omega V^{2}+2N\Gamma V\right)\int\left\Vert \nabla G_{z_{0}}\right\Vert ^{2}\rho\mathrm{d}V,
\end{multline}

On a sphere, we can always choose a single coordinate chart covering
$\Sigma\setminus\left\{ z_{0}\right\} $, in which case $G_{z_{0}}=\frac{1}{2}K$
where $K$ is the local K\"ahler potential defined by $\partial\bar{\partial}K=\frac{\sqrt{g}}{2V}$. Then a simple
computation in complex coordinates leads to
\begin{equation}
    L\left(z,w\right)=2V\Re\frac{\frac{\partial K}{\sqrt{g}}(z)-\frac{\partial K}{\sqrt{g}}(w)}{\bar{z}-\bar{w}}.
\end{equation}
%which is an integral kernel related to the so called \textit{loop
%operator} in the context of Conformal Field Theory \cite[(5.7)]{ZabroWieg06}.
Note that in the limit of the infinite plane, $VK\to\frac{1}{2}\left|z\right|^{2}$
and $L(z,w)\to1$, so we get back the original sum rule (\ref{eq: old sum rule})
(provided that $\Omega V\gg\Gamma N$ in the limit).

\begin{rem*}
    The sum rule in the plane can be alternatively derived from the Hamiltonian by performing a uniform infinitesimal dilatation $z_{i}\to e^{\sigma}z_{i}$ (or $\rho(z)\to e^{-2\sigma}\rho\left(e^{-\sigma}z\right)$ in the continuous formulation) and equating the derivative $\restr{\frac{\partial H}{\partial\sigma}}{\sigma=0}$ to zero. On a closed surface, the necessary transformation is \textit{not} a dilatation of the surface. Instead all vortices need to be displaced along the vector field $\nabla^{\mu}G_{z_{0}}$, i.e.\ roughly towards or away from the large vortex at $z_{0}$. This displacement still looks like a dilatation locally since the divergence of this vector field is constant.
\end{rem*}

\section{Variation of the Green function}\label{sec: Appendix gen}

First, let $\mathring{g}_{\mu\nu}$
be a \textit{reference metric of a constant curvature} that is conformal
to $g_{\mu\nu}$. We denote its total volume by $\mathring{V}$, its covariant derivative by $\mathring\nabla$, its Green function by $\mathring{G}(x,y)$, and so on.

Furthermore, for any conformal deformation $\mathring{g}_{\mu\nu}\to g_{\mu\nu}$
we introduce a potential $K$ (not to be confused with the local Kähler potential we used above) defined up to an additive constant by
\begin{equation}
    -\Delta K=\frac{2}{\mathring{V}}\frac{\sqrt{\mathring{g}}}{\sqrt{g}}-\frac{2}{V},\label{eq: Kahler definition}
\end{equation}
where $\Delta$ is the Laplace-Beltrami operator of the final metric.

The main instruments for our derivation
are the formulas for the metric variations of the Green and the Robin
functions.
\begin{prop}\label{prop: variation of Green}
For any genus, the variation of the Green function $G\left(x,y\right)$
with respect to the metric, defined as $\delta G\left(x,y\right)=\int\frac{\delta G\left(x,y\right)}{\delta g^{\mu\nu}(z)}\delta g^{\mu\nu}(z)\mathrm{d}V(z)$,
is
\begin{multline}
\frac{\delta G\left(x,y\right)}{\delta g^{\mu\nu}\left(z\right)}=-\left[\nabla_{(\mu}G_{x}\nabla_{\nu)}G_{y}-\frac{1}{2}g_{\mu\nu}\nabla_{\rho}G_{x}\nabla^{\rho}G_{y}\right]+\\
+\frac{1}{2V}g_{\mu\nu}\left(G_{x}+G_{y}\right).\label{eq: G(x,y) variation}
\end{multline}
Here, all omitted arguments are $z$.
\end{prop}

This is nothing but a covariant form of the well known Hadamard variational
formula (more often presented as a variation of the boundary of a
domain) or Schiffer's interior variation formula \cite[(7.8.16)]{SchifSpen54}.
The special case of conformal variations (trace part) is often used
in string theory \cite[(2.87)]{DHokPhon88}, and the quasi-conformal
variation (traceless part) can also be found in \cite{Sontag75,Taniguchi89}\cite[(10)]{KuzenSolov91}. The Green function is modular invariant, which is why this formula is universal for surfaces of any genus.

Any purely geometric functional $A\left[g\right]$, i.e.\ a diffeomorphism invariant functional dependent only on the metric, has a divergenceless
metric variation. Namely, for any
infinitesimal diffeomorphism given by a vector field $\xi^{\mu}$,
the resulting variation must vanish:
\begin{equation}
\delta_\xi A\left[g\right]=-\int\frac{\delta A\left[g\right]}{\delta g^{\mu\nu}(z)}\nabla^{(\mu}\xi^{\nu)}\mathrm{d}V(z)=\int\xi^{\nu}\nabla_{\mu}\frac{\delta A\left[g\right]}{\delta g^{\mu\nu}}\mathrm{d}V=0,
\end{equation}
therefore $\nabla_{\mu}\frac{\delta A\left[g\right]}{\delta g^{\mu\nu}(z)}=0$.
If $A\left[g\right](x_{1},\ldots,x_{r})$ is also a function of several
points, we can similarly derive the Ward-type identity
\begin{equation}
\nabla_{\mu}\frac{\delta A\left[g\right](x_{1},\ldots,x_{r})}{\delta g^{\mu\nu}(z)}=\frac{1}{2}\sum_{i=1}^{r}\delta_{x_{i}}\left(z\right)\nabla_{\nu}A\left[g\right]\left(\ldots,x_{i-1},z,x_{i+1},\ldots\right),\label{eq: Ward}
\end{equation}
where all derivatives act on $z$. For instance, using (\ref{eq: G(x,y) variation}),
we can instantly verify that 
\begin{equation}
\nabla^{\mu}\frac{\delta G\left(x,y\right)}{\delta g^{\mu\nu}(z)}=\frac{1}{2}\delta_{x}(z)\nabla_{\nu}G\left(z,y\right)+\frac{1}{2}\delta_{y}(z)\nabla_{\nu}G\left(z,x\right).
\end{equation}
The Ward identity (\ref{eq: Ward}) is also helpful in computing
these variations in two dimensions. Namely, conformal variations are
often known or easy to compute, which gives the trace of the variation.
The Ward identity then reduces the problem to finding a traceless
symmetric tensor with a given divergence, which is a $\bar{\partial}$-problem
in local coordinates (since $\nabla^{\mu}f_{\mu\nu}=2e^{-2\sigma}\bar{\partial}_zf_{zz}$
for a traceless symmetric tensor $f_{\mu\nu}$). Such a problem has a unique solution
up to a holomorphic quadratic differential, and in particular unique on a sphere.

\section{Proofs of main results\label{sec: Appendix g=0}}

On genus zero, the Robin function satisfies 
\[-\Delta G^R=\frac{R}{4\pi}-\frac{2}{V}.\label{laplacian of Robin on g=0}\]
The following key relationships between the curvature, the conformal factor $g_{\mu\nu}=e^{2\sigma}\mathring{g}_{\mu\nu}$, and the potential (\ref{eq: Kahler definition}) hold on genus zero and follow
from (\ref{eq: Kahler definition}, \ref{laplacian of Robin on g=0}):
\begin{gather}
U=\left(-\Delta\right)^{-1}\frac{R}{4\pi}=(K-\wb{K})+\frac{\sigma-\wb{\sigma}}{2\pi},\\
G^{R}=U+C,\quad C=C\left[g\right]=\mathrm{const}.\label{eq: relations for G^R}
\end{gather}
% Moreover, it is known \cite{Steiner05,FerKleZel12} that if the K\"ahler
% potentials are chosen so that their integrals vanish, then
% \begin{align}
% C\left[g\right]+\frac{1}{2\pi}\left(\gamma-\ln2\right) & =\overline{G^{R}}=\lim_{s\to1}\left(\frac{1}{V}\zeta(s)-\frac{1}{4\pi\left(s-1\right)}\right),
% \end{align}
% where $\zeta(s)$ is the zeta function of the operator $-\Delta$.
The bars denote averaging over the surface. 
Therefore, at least on genus zero, we can easily vary $G^{R}$ once we know how
to vary the Green function, the scalar curvature, and the constant
$C$. However, to obtain a universal answer for any genus, we will follow a different tactic: we will make an educated guess of the answer and prove it by verifying the trace and the divergence. But first we need a technical result that is at the heart of coarse-graining.

From now on, we use $\left[a_{\mu\nu}-\mathrm{tr}\right]$ as a shorthand
notation for the traceless part $a_{\mu\nu}-\frac{1}{2}g_{\mu\nu}g^{\rho\lambda}a_{\rho\lambda}$
(unambiguous no matter which of the conformal metrics we use).

\begin{lem}[Generalization of algebraic identity (19) in \cite{WiegmAban14}]
    On a Riemann surface of genus $\mathsf{g}$ with a metric $g$ of total volume $V$ and Green's function $G$, fix a point $x$ and consider the following traceless symmetric tensor as a functional of the metric:
    \begin{equation}
        \Phi_{x,\mu\nu}\coloneqq-\left[\nabla_{\mu}G_{x}\nabla_{\nu}G_{x}-\mathrm{tr}\right]+\frac{1}{4\pi}\left[\nabla_{\mu}\nabla_{\nu}G_{x}-\mathrm{tr}\right],
    \end{equation}
    or in complex conformal coordinates
    \begin{equation}
        \Phi_{x,zz}(z)=-\left(\partial G_{x}\right)^{2}+\frac{1}{4\pi}\partial^{2}G_{x}-\frac{1}{4\pi}\left(\partial\ln\sqrt{g}\right)\partial G_{x}.\label{Phi in complex}
    \end{equation}
    Then
    \begin{enumerate}
        \item Denoting $\mathring{\Phi}_{x,\mu\nu}=\Phi_{x,\mu\nu}[\mathring{g}]$ for the reference metric $\mathring{g}$ of constant curvature, we claim that $\mathring{\Phi}_{\mu\nu}$ is locally integrable and
        \begin{equation}
            \mathring{\nabla}^{\mu}\mathring{\Phi}_{x,\mu\nu}=\frac12\mathring\delta_x\nabla_\nu \mathring{G}^R-\frac{\mathsf{g}}{\mathring{V}}{\nabla}_{\nu}\mathring{G}_{x}.\label{eq: divergence of Phi^circ}
        \end{equation}
        Moreover, $\mathring{\Phi}_{x,\mu\nu}=0$ in the infinite plane and
        on the round sphere. The flat torus case will be stated in (\ref{Phi on flat torus}).
        \item If $g=e^{2\sigma}\mathring{g}$, and $K$ is the potential (\ref{eq: Kahler definition}), then $\Phi_{x,\mu\nu}$ is locally integrable and
        \begin{multline}
            \Phi_{x,\mu\nu}-\mathring{\Phi}_{x,\mu\nu}=-\left[\nabla_{(\mu}G_{x}\nabla_{\nu)}\left(K+\frac{\sigma}{2\pi}\right)-\mathrm{tr}\right]+\\
            +\frac{1}{2V}f_{\mu\nu},
        \end{multline}
        where 
        \begin{equation}
            \frac{1}{V}f_{\mu\nu}=\left[\frac{1}{2}{\nabla}_{\mu}K{\nabla}_{\nu}K-\mathrm{tr}\right]+\left[\frac{1}{4\pi}\mathring{\nabla}_{\mu}{\nabla}_{\nu}K-\mathrm{tr}\right].\label{eq: f_munu formula}
        \end{equation}
        In addition, its divergence is
        \begin{equation}
            \nabla^{\mu}\Phi_{x,\mu\nu}=\frac{1}{2}\delta_{x}\nabla_{\nu}G^R+\left(\frac{R}{8\pi}-\frac{1}{V}\right)\nabla_{\nu}G_{x}.\label{eq: divergence of Phi}
        \end{equation}
    \end{enumerate}
\end{lem}

\begin{proof}
    Since locally 
    \begin{multline}
        \mathring{G}_{x}(z)=-\frac{1}{2\pi}\ln\left|z-x\right|-\frac{1}{4\pi}\ln\sqrt{\mathring{g}(z)}-\frac{1}{4\pi}\ln\sqrt{\mathring{g}(x)}+\\+\frac12 \mathring{G}^R(z)+\frac12 \mathring{G}^R(x)+\mathcal{O}(z-x)
    \end{multline}
    with a smooth error term, the absence
    of second order poles in $\mathring{\Phi}_{\mu\nu}$ follows from
    the identity 
    \begin{equation}
        -\left(\frac{1}{2\pi}\partial\ln\left|z-x\right|\right)^{2}-\frac{1}{4\pi}\partial^{2}\left(\frac{1}{2\pi}\ln\left|z-x\right|\right)=0.
    \end{equation}
    Therefore $\mathring{\Phi}_{\mu\nu}$ is well-defined as a distribution
    and in turn has a well-defined divergence. The value of the divergence $\frac{2}{\sqrt{g}}\bar\partial \mathring{\Phi}_{x,zz}$
    away from $x$ is immediately found to be the second term in (\ref{eq: divergence of Phi^circ}).
    The rest has to be concentrated at $z=x$. The only delta-like term
    can come from the first order pole in $\mathring{\Phi}_{zz}$ and by substitution of the above short distance expansion of the Green function into (\ref{Phi in complex})
    we find it to be equal to $\mathring{\delta}_{x}(z)\frac{1}{2}\partial \mathring{G}^R$.

    The second part of the proposition follows from the transformation
    rule $G\left(x,y\right)=\mathring{G}\left(x,y\right)+\frac{1}{2}K(x)+\frac{1}{2}K(y)+\mathrm{const}$
    and the transformation rule for the Christoffel symbols
    \begin{gather}
        \left(\nabla_{\mu}-\mathring{\nabla}_{\mu}\right)a_{\nu}=-C_{\mu\nu}^{\lambda}a_{\lambda},\\
        C_{\mu\nu}^{\lambda}=\delta_{\mu}^{\lambda}\partial_{\nu}\sigma+\delta_{\nu}^{\lambda}\partial_{\mu}\sigma-g_{\mu\nu}g^{\lambda\rho}\partial_{\rho}\sigma,\;\nabla^{\mu}C_{\mu\nu}^{\lambda}=\delta_{\nu}^{\lambda}\Delta\sigma.
    \end{gather}
    An alternative field-theoretical derivation of these identities, as well as a more explicit formula for $\Phi_{x,zz}$, can be found in \cite{KuzenSolov91}.
\end{proof}
This Lemma defines $\nabla_\mu G_x\nabla_\nu G_x$ as a distribution: the double derivative $\nabla_\mu\nabla_\nu G_x$ can be understood as a second derivative of a regular distribution (which is always well-defined), and $\Phi_{\mu\nu}$ is a regular distribution too. This gives a precise distributional meaning to the vortex stress $P_{\mu\nu}$ and finalizes the statement of Proposition \ref{prop: Stress}.

\begin{lem}\label{prop: variation of Robin}
    On surfaces of any genus,
    \begin{multline}
        \frac{\delta G^{R}(x)}{\delta g^{\mu\nu}(z)}=-\left[\nabla_\mu G_x\nabla_\nu G_x-\tr\right]+h_{\mu\nu}+\frac{1}{2}g_{\mu\nu}\left(-\frac{\delta_{x}}{4\pi}+\frac{2}{V}G_x\right),\label{eq: variation of Robin after Phi}
    \end{multline}
    where $h_{\mu\nu}(z)$ is a (possibly $x$-dependent) holomorphic quadratic differential (i.e.\ a divergenceless and traceless symmetric tensor). On genus zero surfaces $h_{\mu\nu}=0$ and on genus one surfaces $h_{\mu\nu}$ is $x$-independent. The first term in the variation has to be understood as the distribution equal to $\Phi_{\mu\nu}-(\nabla_\mu\nabla_\nu G_x-\tr)/4\pi$.
\end{lem}
\begin{proof}
    Using the previous Lemma, we directly check the Ward identity
    \[\nabla^\mu \frac{\delta G^R(x)}{\delta g^{\mu\nu}}=\frac12\delta_x \nabla_\nu G^R.\]
    Moreover, the trace part of the variation (corresponding to conformal deformations of the metric) follows directly from the definition of $G^R$ and Proposition \ref{prop: variation of Green}:
    \begin{equation}
        -g^{\mu\nu}\frac{\delta G^{R}(x)}{\delta g^{\mu\nu}(z)}=\frac{1}{2}\frac{\delta G^{R}(x)}{\delta\sigma(z)}=-\frac{2}{V}G\left(x,z\right)+\frac{1}{4\pi}\delta_{x}(z).\label{eq: trace variation of G^R general}
    \end{equation}
    Therefore $\delta G^R/\delta g^{\mu\nu}$ can differ from the stated result only by a traceless divergenceless term. On genus zero surfaces such a term has to be zero. We show that this term also has to be $x$-independent (and therefore irrelevant to the physics of vortices) at least on genus one surfaces in Appendix \ref{sec: Appendix g=1}.
\end{proof}

% The above Lemma is helpful since the divergence of $\Phi_{x,\mu\nu}$
% cannot be naively computed term-by-term because $\bar{\partial}\left(\partial G_{x}\right)^{2}$
% cannot be expanded via the chain rule in any distributional sense.
% However, $\Phi_{x,\mu\nu}$ is a locally integrable expression in
% general, and after the cancellation of the singularities of the type
% of second order poles its divergence is easily computed. This is akin
% to trying to compute $\bar{\partial}\frac{1}{z^{2}}$: it seems to
% make no sense since the result $\frac{2}{z}\cdot\pi\delta\left(z\right)$
% obtained from a naive application of the chain rule is undefined,
% but at the same time it can be easily resolved by $\bar{\partial}\frac{1}{z^{2}}=-\bar{\partial}\partial\frac{1}{z}=-\partial\bar{\partial}\frac{1}{z}=-\pi\partial\delta\left(z\right)$,
% which is a well defined distribution. Therefore $\Phi_{z_{i},\mu\nu}$
% is the contribution to the stress tensor that needs to be added after
% terms of the form $\frac{1}{4\pi}\nabla_{\mu}\nabla_{\nu}G\left(z,z_{i}\right)$
% are replaced by ``diagonal'' terms of the form $\nabla_{(\mu}G\left(z,z_{i}\right)\nabla_{\nu)}G\left(z,z_{i}\right)$. And even though doing so
%  makes it slightly harder to make sense of the resulting expression
% as a distribution, it eventually leads to great simplifications.

\begin{lem}\label{prop: variation of U}
    The metric variation of the curvature potential (\ref{eq: def of curvature potential}) is 
    \begin{multline}
        \frac{\delta U(x)}{\delta g^{\mu\nu}(z)}
        =-\frac{1}{4\pi}\left[\nabla_{\mu}\nabla_{\nu}G_{x}-\mathrm{tr}\right]-\left[\nabla_{(\mu}G_{x}\nabla_{\nu)}U-\mathrm{tr}\right]+\\
        +\frac{1}{2}g_{\mu\nu}\left(-\frac{\delta_{x}}{4\pi}+\frac{\chi}{V}G_{x}+\frac{1}{V}U+\frac{1}{4\pi V}\right)\label{eq: variation of curvature potential}
    \end{multline}
\end{lem}
\begin{proof}
    The variation of $R$ is standard, 
    \begin{align}
    \delta\left(R\mathrm{d}V\right) & =\nabla_{\mu}k^{\mu}\mathrm{d}V,\\
    k^{\mu} & =-\nabla_{\nu}\delta g^{\mu\nu}+g_{\lambda\rho}g^{\mu\nu}\nabla_{\nu}\delta g^{\lambda\rho}.
    \end{align}
    This together with Proposition \ref{prop: variation of Green} gives the result, since $U(x)=\frac{1}{4\pi}\int G(x,z)R(z)\dd V(z)$.
\end{proof}

\begin{proof}[Proof of Proposition \ref{prop: Stress}]
Combining all of the above lemmas we find that
the variation of the total Hamiltonian (\ref{eq: discrete Hamiltonain (m)})
gives the stress tensor (\ref{eq: stress P}). There, the result is intentionally written in terms of the full stream function $\psi$ instead of just $\Gamma(-\Delta)^{-1}\rho$, which leads to the extra term
\begin{multline}
    \frac{1}{2m\Gamma^2}o_{\mu\nu}=\left[\nabla_\mu U\nabla_\nu U+\frac{1}{2\pi}\nabla_\mu\nabla_\nu U-\tr\right]+\\
    +\frac12 g_{\mu\nu}\left(\frac{R}{8\pi^2}-\frac{2\chi}{V}U-\frac{N}{4\pi V}\right).
\end{multline}
This term is divergenceless and can be eliminated by an addition to the Hamiltonian of a constant dependent only on the metric. On higher genus surfaces, $o_{\mu\nu}$ will also absorb $h_{\mu\nu}$ coming from the variation of $G^R$. The proof that $h_{\mu\nu}$ and therefore $o_{\mu\nu}$ remains $z_i$-independent on genus one surfaces will be given in Appendix \ref{sec: Appendix g=1}.
\end{proof}

%The very last step in this proof is the only one where the quantization requirement $\Gamma_i=\Gamma$ really comes into play: we need to be able to rewrite $\sum\Gamma_i^2\nabla_\mu\nabla_\nu G_{z_i}$ in terms of the vorticity $\sum \Gamma_i\delta_{z_i}$. The only possible generalization is vortices of two signs, $\Gamma_i=\pm \Gamma$, in which case $\sum\Gamma_i^2\nabla_\mu\nabla_\nu G_{z_i}=\Gamma (-\Delta)^{-1}|\omega|$.

\begin{lem}
    Let $H_N\left(\{z_i\}\right)$ be a symmetric function of the positions $z_i$ of $N$ identical particles, such that in the coarse-graining limit it is represented as a functional $H_{\rm CG}[\rho]$ of the coarse-grained density $\rho$. We also assume that $H_N$ and $H_{\rm CG}$ are known functionals of the metric $g_{\mu\nu}$ on the space on which the particles reside. Then
    \begin{equation}
        \rho\nabla_{\mu}\frac{\delta H_{\rm CG}}{\delta\rho}=2\nabla^{\nu}\frac{\delta H_{\rm CG}}{\delta g^{\mu\nu}}.
    \end{equation}
\end{lem}
\begin{proof}
This is nothing but a manifestation of diffeomorphism invariance,
i.e.\ just a continuous form of the Ward identity (\ref{eq: Ward}).
It can be proven by considering an infinitesimal displacement of all
particles along a vector field $\xi^{\mu}$, given by the Lie derivative
of the density
\begin{align}
    \rho\left(x;\left\{ z_{i}\right\} \right) & =\sum_{i}\delta\left(x,z_{i}\right),\\
    \delta_{\xi}\rho\left(x;\left\{ z_{i}\right\} \right) & =\sum_{i}\xi^{\mu}\left(z_{i}\right)\partial_{z_{i},\mu}\delta\left(x,z_{i}\right).
\end{align}
% Here we can treat $\delta\left(x,z_{i}\right)$ as a scalar density
% of weight $\frac{1}{2}-p$ as a function of $x$ and a density of
% weight $\frac{1}{2}+p$ as a function of $z_{i}$ for any real $p$
% (for densities $f$ of weight $h$, $f/\sqrt{g}^{h}$ is a true scalar
% and $\nabla_{\mu}f=\partial_{\mu}f-h\left(\partial_{\mu}\ln\sqrt{g}\right)f$).
Using the definition of the delta function and the integration properties
of covariant and Lie derivatives along $\xi$ (denoted $\nabla_{\xi(x)}=\xi^{\mu}(x)\nabla_{x,\mu}$
and $\partial_{\xi(z)}=\xi^{\mu}(z)\partial_{z,\mu}$, respectively),
we can verify the identity
% \begin{equation}
% \begin{cases}
%     \int f(x)\nabla_{\xi(x)}\delta(x,z)\mathrm{d}V(x)=-\partial_{\xi}f(z)-f(z)\nabla_{\mu}\xi^{\mu}(z),\\
%     \int f(x)\partial_{\xi(z)}\delta(x,z)\mathrm{d}V(x)=\partial_{\xi}f(z),
% \end{cases}
% \end{equation}
% implying 
\begin{multline}
    \delta_{\xi}\rho=\sum_{i}\partial_{\xi(z_{i})}\delta\left(x,z_{i}\right)=-\sum_{i}\left(\nabla_{\xi(x)}\delta_{z_{i}}(x)+\delta_{z_i}(x)\cdot\nabla_{\mu}\xi^{\mu}\right)=\\
    =-\sum_i\nabla_{\mu}\left(\xi^{\mu}(x)\delta_{z_i}(x)\right)=-\nabla_{\mu}\left(\rho\xi^{\mu}\right).
\end{multline}
Now we let this identity \emph{define} the variations of the coarse-grained density.  Then for any Hamiltonian
$H\left[\rho,g\right]$
\begin{align}
    \delta_{\xi}H & =\int\frac{\delta H}{\delta\rho}\delta_{\xi}\rho\,\mathrm{d}V=\int\rho\xi^{\mu}\nabla_{\mu}\frac{\delta H}{\delta\rho}\,\mathrm{d}V.
\end{align}
On the other hand, we have the diffeomorphism-induced metric variation
$\delta g^{\mu\nu}=-\nabla^{\mu}\xi^{\nu}-\nabla^{\nu}\xi^{\mu}$,
giving the metric variation
\begin{gather}
    \delta_{g}H=\frac{1}{2}\int P_{\mu\nu}\delta g^{\mu\nu}\mathrm{d}V=\int\xi^{\mu}\nabla^{\nu}P_{\nu\mu}\mathrm{d}V.
\end{gather}
Diffeomorphism-invariance is exactly the statement that $\delta_\xi H=\delta_g H$. Comparing the two variations for all $\xi^{\mu}$, we get $\nabla^{\nu}P_{\nu\mu}=\rho\nabla_{\mu}\frac{\delta H}{\delta\rho}$.
\end{proof}

\begin{proof}[Proof of (\ref{coarse-grained v in terms of H})]
    The formula for the coarse-grained flow follows directly from the previous lemma. The velocity of the particles, as prescribed by the symplectic structure, is $\Gamma\rho\wh v^{\ast}_\mu=-\rho\nabla_{\mu}\frac{\delta H}{\delta\rho}$, which gives the result after dividing by $\rho$.
\end{proof}

\begin{proof}[Proof of Proposition \ref{prop: Coarse-grained flow}]
    On the discrete level, the divergence of $P_{\mu\nu}$ has to be computed carefully by referring to $\Phi_{\mu\nu}$, and gives the exact Kirchhoff equations. But now we use the power of the representation of $P_{\mu\nu}$ in terms of $\rho$ to coarse-grain. Assuming that the flow has been coarse-grained ($u^\mu$ no longer has singularities), we compute the divergence directly using the final expression for $P_{\mu\nu}$:
    \begin{equation}
        \nabla^{\mu}P_{\mu\nu}=\Gamma^2\rho\nabla(-\Delta)^{-1}\rho-m\Gamma^{2}\rho\nabla U-\frac{\Gamma^{2}}{8\pi}\nabla\rho,
    \end{equation}
    We can divide both sides by the coarse-grained $\rho$, proving the general coarse-graining formula. The stated mean-field
    equation is obtained by taking a curl of $v^\mu$ and equating it to zero.
\end{proof}

\section{Torus\label{sec: Appendix g=1}}

In this section we rederive the expression for the stress tensor on genus one surfaces, and in particular we show that the term $o_{\mu\nu}$ remains $z_i$-independent. Tori are special among all closed surfaces in that we
can in fact obtain a formula for the metric variation of $\sigma(x)$,
$K(x)$ etc. They have $\mathsf{g}=1$ and $\chi=0$, so the general relation
\begin{gather}
    U=\left(-\Delta\right)^{-1}\frac{R}{4\pi}=\frac{\sigma-\wb{\sigma}}{2\pi}+\frac{\chi}{2}(K-\wb{K}),\label{eq: curvature potential as sigma and K}
\end{gather}
simply states that $\sigma=2\pi U$, therefore
we already know how to vary it. Finally, on a torus, just like on a sphere, $G^{R}=\frac{\sigma}{2\pi}+K+C[g]$, whose variation is now computable since  $K=\frac{2}{\mathring{V}}\left(-\Delta\right)^{-1}e^{-2\sigma}$ is expressed in terms of $\sigma$.
The only remaining unknown will be the traceless part of the variation
of $C[g]$ (the total trace of the variation is already
known from (\ref{eq: trace variation of G^R general})). However,
this part will be a smooth tensor independent of $x$ that is not physically relevant anyway.

Since the variation of $\sigma=2\pi U$ follows from (\ref{eq: variation of curvature potential})
with $\chi=0$,
\begin{multline}
\frac{\delta\left(2\sigma(x)\right)}{\delta g^{\mu\nu}(z)}=-\left[\nabla_{(\mu}G_{x}\nabla_{\nu)}\left(2\sigma\right)-\frac{1}{2}g_{\mu\nu}\nabla_{\rho}G_{x}\nabla^{\rho}\left(2\sigma\right)\right]-\\
-\left[\nabla_{\mu}\nabla_{\nu}G_{x}-\frac{1}{2}g_{\mu\nu}\Delta G_{x}\right]+\frac{1}{2}g_{\mu\nu}\left(-\delta_{x}+\frac{1}{V}+\frac{2\sigma}{V}\right),
\end{multline}
we can derive the total variation of $K$ on
a torus,
\begin{multline}
\frac{\delta K(x)}{\delta g^{\mu\nu}(z)}=-\left[\nabla_{(\mu}G_{x}\nabla_{\nu)}K-\mathrm{tr}\right]+\\
+\left[\frac{2}{\mathring{V}}\mathring{\nabla}_{\mu}{\nabla}_{\nu}\left(\left(-\mathring{\Delta}\right)^{-1}{G}_x\right)-\mathrm{tr}\right]+\frac{1}{2}g_{\mu\nu}\left(\frac{2}{V}G_{x}+\frac{1}{V}K\right),
\end{multline}
(note no circle above $G_x$). Once again, all omitted free arguments here are $z$ and all derivatives are w.r.t.\ $z$. 
It is now easy to notice that on a flat torus
\[
    \mathring{\Phi}_{x,\mu\nu}=\left[\frac{2}{\mathring{V}}\mathring{\nabla}_{\mu}{\nabla}_{\nu}\left(\left(-\mathring{\Delta}\right)^{-1}\mathring{G}_x\right)-\mathrm{tr}\right]+h_{\mu\nu},\label{Phi on flat torus}
\]
where $h_{\mu\nu}$ is a holomorphic quadratic differential ($h_{zz}=\mathrm{const}$). By translation invariance, $h_{\mu\nu}$ has to be $x$-independent.

Considering that we know what the exact trace part of $\delta G^R/\delta g^{\mu\nu}$
should be from (\ref{eq: trace variation of G^R general}), all of
this adds up to the same variation of $G^{R}$ as on a sphere up to the uninteresting $h_{\mu\nu}$:
% \begin{multline}
% \frac{\delta G^{R}(x)}{\delta g^{\mu\nu}(z)}=-\frac{1}{4\pi}\left[\nabla_{\mu}\nabla_{\nu}G_{x}-\mathrm{tr}\right]-\\
% -\left[\nabla_{(\mu}G_{x}\nabla_{\nu)}G^{R}-\mathrm{tr}\right]+\\
% +\left[\frac{2}{\mathring{V}}\nabla_{(\mu}\left(2\sigma\right)\nabla_{\nu)}\left(\left(-\Delta\right)^{-1}\left(e^{-2\sigma}G_{x}\right)\right)-\mathrm{tr}\right]+\\
% +\left[\frac{2}{\mathring{V}}\nabla_{\mu}\nabla_{\nu}\left(\left(-\Delta\right)^{-1}\left(e^{-2\sigma}G_{x}\right)\right)+\mathrm{tr}\right]+\\
% -\frac{K(x)}{2}\left[\nabla_{(\mu}K\nabla_{\nu)}\left(2\sigma\right)-\mathrm{tr}\right]-\frac{K(x)}{2}\left[\nabla_{\mu}\nabla_{\nu}K-\mathrm{tr}\right]+\\
% +\frac{1}{2V}f^\prime_{\mu\nu}(z)+\frac{1}{2}g_{\mu\nu}\left(-\frac{1}{4\pi}\delta_{x}+\frac{2}{V}G_{x}
% %-\frac{\sigma}{2\pi V}-\frac{1}{4\pi V}
% \right),
% \end{multline}
\begin{gather}
    \frac{\delta G^{R}(x)}{\delta g^{\mu\nu}(z)}=-\left[\nabla_\mu G_x\nabla_\nu G_x-\tr\right]+h_{\mu\nu}+\frac{1}{2}g_{\mu\nu}\left(-\frac{\delta_{x}}{4\pi}+\frac{2}{V}G_x\right). \label{torus variation}
\end{gather}
where
% \begin{multline}
%     o^{\prime}_{\mu\nu}=\left[\frac{1}{\mathring{V}}\mathring{\nabla}_\mu\nabla_\nu\left((-\Delta)^{-1}K\right)-\tr\right]-\frac{1}{2V}f_{\mu\nu}+h_{\mu\nu}+\\+\frac{1}{2V}g_{\mu\nu}\left(K+\frac{\sigma}{2\pi}\right),
% \end{multline}
as usual the first term has to be understood as $\Phi_{\mu\nu}-(\nabla_\mu\nabla_\nu G_x-\tr)/4\pi$.
% \begin{equation}
% \nabla^{\mu}f^\prime_{\mu\nu}(z)=\nabla_{\nu}\left(K+\frac{\sigma}{\pi}\right).
% \end{equation}

This implies that the final form of the stress tensor given in Proposition \ref{prop: Stress} remains correct with a $z_i$-independent $o_{\mu\nu}$. The question of whether $o_{\mu\nu}$ is $z_i$-independent on surfaces of higher genus as well remains open.

%\bigbreak
%\PRLsep

\end{document}